\documentclass[journal,onecolumn]{IEEEtran}

\usepackage{amsthm}
\usepackage{amsmath}
\usepackage{mathrsfs}
\usepackage{amssymb} 
\usepackage{bbm,dsfont} 
\usepackage{multirow} 
\usepackage{units}
\usepackage{stmaryrd}   
\usepackage{verbatim}
\usepackage{enumerate} 
\usepackage[table]{xcolor}
\usepackage{slashbox} 
\usepackage{ wasysym }
\usepackage{tikz}
\usetikzlibrary{calc}
\usepackage{tabularx}

\usepackage{lipsum}
\usepackage{multicol}
\usepackage{float}

\usepackage{ifthen}

\usepackage{xcolor}
\usepackage{hyperref} 

\usepackage{url}

\usepackage{breakurl}

\usepackage
{hyperref} 
\hypersetup{
    colorlinks,
    linkcolor={blue!80!black},
    citecolor={green!50!black},
    urlcolor={blue!80!black}
}

\newcommand\blfootnote[1]{%
  \begingroup
  \renewcommand\thefootnote{}\footnote{#1}%
  \addtocounter{footnote}{-1}%
  \endgroup
}

\usepackage{graphicx}
\graphicspath{{images/}}

\title{The Arbitrarily Varying Gaussian Relay Channel with Sender Frequency Division}
\author{\IEEEauthorblockN{Uzi Pereg and Yossef Steinberg}
\IEEEauthorblockA{Department of Electrical Engineering\\
Technion, Haifa 32000, Israel.\\
Email: {\tt uzipereg@campus.technion.ac.il}, {\tt ysteinbe@ee.technion.ac.il}
 }
}

\usepackage[backend=biber,sorting=nyt,maxbibnames=10]{biblatex}
\ifdefined\bibstar\else\newcommand{\bibstar}[1]{}\fi
\renewbibmacro{in:}{} 										

\definecolor{light-gray}{gray}{0.8}

\definecolor{dark-gray}{gray}{0.3}
\usepackage{accents}
\newlength{\dhatheight}

\newcommand{\bieee}{\begin{IEEEeqnarray}{rCl}}
\newcommand{\eieee}{\end{IEEEeqnarray}}
\newcommand{\prob}[1]{\Pr\left(#1\right)}
\newcommand{\given}{\mid}
\newcommand{\cprob}[2]{\Pr\left(#1\given #2\right)}
\newcommand{\E}{\mathbb{E}}
\newcommand{\var}{\mathbb{V}\mathrm{ar}}
\newcommand{\eps}{\varepsilon}

\newcommand{\norm}[1]{\left\lVert#1\right\rVert}

\newcommand{\ie}{\emph{i.e.} }
\newcommand{\eg}{\emph{e.g.} }

\newcommand{\cf}{\emph{cf.} }

\newcommand{\xvec}{\mathbf{x}}
\newcommand{\yvec}{\mathbf{y}}
\newcommand{\zvec}{\mathbf{z}}
\newcommand{\fvec}{\mathbf{f}}

\newcommand{\svec}{\mathbf{s}}
\newcommand{\avec}{\mathbf{a}}

\newcommand{\cvec}{\mathbf{c}}
	
\newcommand{\vvec}{\mathbf{v}}	
\newcommand{\Xvec}{\mathbf{X}}
\newcommand{\Yvec}{\mathbf{Y}}
\newcommand{\Zvec}{\mathbf{Z}}
\newcommand{\Svec}{\mathbf{S}}
		
\newcommand{\tm}{\widetilde{m}}	
\newcommand{\tM}{\widetilde{M}}

\newcommand{\tY}{\widetilde{Y}}

\newcommand{\tx}{\tilde{x}}

\newcommand{\ty}{\tilde{y}}

\newcommand{\tf}{\widetilde{f}}
\newcommand{\tg}{\widetilde{g}}

\newcommand{\gnu}{\tg}
\newcommand{\bR}{\tR}
\newcommand{\oS}{\overline{S}}

\newcommand{\oq}{\overline{q}}

\newcommand{\tR}{\widetilde{R}}

\newcommand{\hm}{\hat{m}}

\newcommand{\hgamma}{\widehat{\gamma}}

\newcommand{\hP}{\hat{P}}
\newcommand{\hM}{\hat{M}}

\newcommand{\Aset}{\mathcal{A}}

\newcommand{\Dset}{\mathcal{D}}
\newcommand{\Fset}{\mathcal{F}}

\newcommand{\Jset}{\mathcal{J}}
\newcommand{\Uset}{\mathcal{U}}

\newcommand{\Qset}{\mathcal{Q}}

\newcommand{\Sset}{\mathcal{S}}
\newcommand{\Wset}{\mathcal{W}}
\newcommand{\Xset}{\mathcal{X}}
\newcommand{\Yset}{\mathcal{Y}}

\newcommand{\Eset}{\mathcal{E}}
\newcommand{\markovC}[1]{%
\begin{tikzpicture}[#1]%
\draw (0,0.3ex) -- (1ex,0.3ex);%
\draw (0.5ex,0.3ex) circle (0.2ex);
\draw[white] (0.2ex,0) -- (0.5ex,0);%
\end{tikzpicture}%
}
\newcommand{\Cbar}{\markovC{scale=2}}

\theoremstyle{remark}	\newtheorem{theorem}{Theorem}
\theoremstyle{remark}	\newtheorem{lemma}[theorem]{Lemma}
\theoremstyle{remark}	\newtheorem{coro}[theorem]{Corollary}
\theoremstyle{remark} \newtheorem{definition}{Definition}
\theoremstyle{remark} 
\theoremstyle{remark} \newtheorem{example}{Example}

\newcommand{\channel}{W_{Y|X,S}}
\newcommand{\avc}{\Wset}																		
\newcommand{\opC}{\mathbb{C}}																
\newcommand{\inC}{\mathsf{C}}															 	
\newcommand{\inR}{\mathsf{R}}

\newcommand{\pSpace}{\mathcal{P}}														

\newcommand{\dM}{\mathsf{M}}															 	

\newcommand{\enc}{\mathrm{f}}																				
\newcommand{\renc}{\mathrm{f}}																				
\newcommand{\dec}{g}																			 	
\newcommand{\code}{\mathscr{C}}															
\newcommand{\gcode}{\mathscr{C}^{\,\Gamma}}									

\newcommand{\cerr}{P_{e|\svec}^{(n)}}													
\newcommand{\err}{P_e^{(n)}}															
\newcommand{\plimit}{\Omega}																			
\newcommand{\tset}{\Aset^{\delta}}													
\newcommand{\qn}{q}
\newcommand{\tQ}{\hat{\Qset}_n}														

\newcommand{\encn}{f^n}																			

\newcommand{\rstarC}{																			  
\, \hspace{-0.3cm} \text{ $$ \mbox{  
\hspace{-0.1cm} 
\small $\star$   
} $$ }
\hspace{-0.25cm}}

\newcommand{\rCav}{\opC^{\rstarC}\hspace{-0.1cm}(\avc)}

\newcommand{\rc}{W_{Y,Y_1|X,X_1,S}}													
\newcommand{\tRYset}{\mathcal{L}}
\newcommand{\RYcompound}{\tRYset^\Qset} 
\newcommand{\avrc}{\tRYset}																			
\newcommand{\RYopC}{\mathbb{C}}																
									
\newcommand{\RYrCcompound}{\RYopC^{\rstarC}\hspace{-0.1cm}(\RYcompound)}

\newcommand{\RYrCav}{\RYopC^{\rstarC}\hspace{-0.1cm}(\avrc)}  

\newcommand{\RYCcompound}{\RYopC(\RYcompound)}
\newcommand{\RYCavc}{\RYopC(\avrc)}

\newcommand{\prc}{W_{Y,Y_1|X,S}}													
\newcommand{\pRYcompound}{\pavrc^\Qset} 
\newcommand{\pRYcompoundP}{\pavrc^{\pSpace(\Sset)}} 

\newcommand{\pavrc}{\tRYset}
									
\newcommand{\pRYrCcompound}{\RYopC^{\rstarC}\hspace{-0.1cm}(\pRYcompound)}
\newcommand{\pRYrCcompoundP}{\RYopC^{\rstarC}\hspace{-0.1cm}(\pRYcompoundP)}
\newcommand{\pRYrCav}{\RYopC^{\rstarC}\hspace{-0.1cm}(\pavrc)}  
\newcommand{\pRYrICav}{\inR_{CS}^{\rstarC}}

\newcommand{\pRYCcompound}{\RYopC(\pRYcompound)}
\newcommand{\pRYCavc}{\RYopC(\pavrc)}
\newcommand{\pRYICcompound}{\inR_{CS}(\pRYcompound)}
\newcommand{\pRYdIRcompound}{\inR_{DF}(\pRYcompound)} 

\newcommand{\pRYdIRavc}{\inR_{DF}^{\rstarC}}

\begin{document}
\maketitle

{}

\begin{abstract} 
We consider the arbitrarily varying Gaussian relay channel with sender frequency division. We determine the random code capacity, and establish lower and upper bounds on the deterministic code capacity. 
It is observed that when the channel input is subject to a low power limit, the deterministic code capacity may be strictly lower than the random code capacity,  and the gap vanishes  as the input becomes less constrained.  

A second model addressed in this paper is the general case of primitive arbitrarily varying relay channels.
We develop lower and upper bounds on the random code capacity, and give conditions under which the deterministic code capacity coincides with the random code capacity, and conditions under which it is lower.
Then, we establish the capacity of the primitive counterpart of the arbitrarily varying Gaussian relay channel with sender frequency division. In this case, the deterministic and random code capacities are the same.
\end{abstract}


\blfootnote{
This work was supported by the Israel Science Foundation (grant No. 1285/16).
}

\section{Introduction}

Recently, there has been a growing interest in the Gaussian relay channel, 
as \eg in \cite{ElGamalZahedi:05p,XueSandhu:07p,
KolteOzgurElGamal:15p,WuOzgur:18p,
WuBarnesOzgur:17c} and references therein. In particular, El Gamal and Zahedi \cite{ElGamalZahedi:05p} introduced the Gaussian relay channel with sender frequency division (SFD), as a special case of a relay channel with orthogonal sender components,  described as follows. 
The transmitter sends a sequence of pairs $\Xvec=(X_i',X_i'')_{i=1}^n$. At time  $i$, the relay receives the 
symbol $Y_{1,i}$, and transmits $X_{1,i}$ based on past received values 
$Y_{1,1},Y_{1,2},\ldots, Y_{1,i-1}$, and the destination decoder receives $Y_i$, with
the following input-output relation,
\begin{align}
&\Yvec_{1} = \Xvec''+\Zvec \,,\nonumber\\
&\Yvec= \Xvec'+\Xvec_{1}+\Svec \,.
\label{eq:GsfdDef}
\end{align}
The transmitter and the relay are subject to input constraints, $\frac{1}{n}\sum_{i=1}^n (X_i'^2+X_i''^2)\leq\plimit$ and $\frac{1}{n}\sum_{i=1}^n X_{1,i}^2\leq\plimit_1$, respectively. 
El Gamal and Zahedi \cite{ElGamalZahedi:05p} determined the capacity of this channel, under the assumption that $\Zvec$ and $\Svec$  are each independent and identically distributed (i.i.d.) according to a given normal distribution, $\mathcal{N}(0,\sigma^2)$ and $\mathcal{N}(0,\theta^2)$, respectively.
The model is especially relevant when the sender and the relay communicate over different frequency bands
\cite{ElGamalKim:11b}. 
%

 In practice,  channel statistics are not necessarily known in exact, and they may even change over time.
This has motivated the study of various arbitrarily varying networks (see \eg\cite{
AhlswedeCai:99p,HofBross:06p,GoldfeldCuffPermuter:16p}).
In particular, this is the case with the Gaussian arbitrarily varying channel (AVC)  without a relay, specified by the relation $\Yvec=\Xvec+\Svec+\Zvec$, where $\Svec$ is a state sequence of unknown joint distribution $F_{\Svec}$, not necessarily independent nor stationary,  and the noise sequence $\Zvec$ is 
i.i.d. $\sim\mathcal{N}(0,\sigma^2)$.
The state sequence can be thought of as if generated by an adversary, or a \emph{jammer}, who randomizes the channel states arbitrarily in an attempt to disrupt communication. 
It is assumed that the user and the jammer are subject to input and state constraints, $\frac{1}{n}\sum_{i=1}^n X_{i}^2\leq\plimit$ and $\frac{1}{n}\sum_{i=1}^n S_i^2\leq\Lambda$ with probability $1$, respectively.
 In \cite{HughesNarayan:87p},  Hughes and Narayan showed that the
random code capacity, \ie the capacity achieved with common randomness, 
 is given by $\inC^{\rstarC}_1=\frac{1}{2}\log(1+\frac{\plimit}{\sigma^2+\Lambda})$.
%
%
%
%
Subsequently,  Csisz{\'{a}}r and Narayan \cite{CsiszarNarayan:91p} showed that the deterministic code capacity, also referred to as simply capacity, demonstrates a dichotomy property. That is, either the capacity  coincides with the random code capacity or else, it is zero. Specifically, the capacity is given by 
\begin{align}
\inC_1=\begin{cases}
\inC^{\rstarC}_1 &\text{if $\Lambda<\plimit$}\,,\\
0 &\text{if $\Lambda\geq\plimit$}\,.
\end{cases}
\label{eq:Gavc1}
\end{align}
It is pointed out in \cite{CsiszarNarayan:91p} that this result is \emph{not} a straightforward consequence of the elegant Elimination Technique \cite{Ahlswede:78p}, used by Ahlswede to establish dichotomy for the AVC without constraints. Although the direct part proof by Csisz{\'{a}}r and Narayan  is 
based on a simple minimum-distance decoder \cite{CsiszarNarayan:91p}, the analysis is a lot more involved compared to \cite{Ahlswede:78p}.

In this work, we study a version of the Gaussian arbitrarily varying  relay channel (AVRC), which is a combination of the Gaussian relay channel with SFD and the Gaussian AVC under input and state constraints.
The channel specified by (\ref{eq:GsfdDef}) is now governed by a state sequence $\Svec$  with an arbitrary joint distribution, which could also give probability mass $1$ to some 
$\svec\in\mathbb{R}^n$, 
and a Gaussian noise sequence $\Zvec$ which is  i.i.d. $\sim\mathcal{N}(0,\sigma^2)$, subject to input and state  constraints, $\frac{1}{n}\sum_{i=1}^n (X_i'^2+X_i''^2)\leq\plimit$, $\frac{1}{n}\sum_{i=1}^n X_{1,i}^2\leq\plimit_1$, and   $\frac{1}{n}\sum_{i=1}^n S_i^2\leq\Lambda$.
The random code capacity is determined following the results in a previous work by the authors \cite{PeregSteinberg:17c4}, which gives  lower and upper bounds on the random code capacity of the general AVRC. Our main contribution in this paper is then to establish lower and upper bounds on the deterministic code capacity, extending the techniques by Csisz{\'{a}}r and Narayan \cite{CsiszarNarayan:91p} to the relay channel. The analysis is thus independent of the results in \cite{PeregSteinberg:17c4}. 
As in the basic scenario in \cite{CsiszarNarayan:91p}, it is observed that when the power limits $\plimit$ and $\plimit_1$ are low, %
   the capacity is below the random code capacity,    and the gap vanishes   as $\plimit$ and $\plimit_1$ increase. 
	
	A second model addressed in this paper is the general case of primitive arbitrarily varying relay channels, where there is a noiseless link between the relay and the receiver of limited capacity \cite{Kim:07c} (see also \cite{Xue:2014p,WuOzgur:18p,AsadiPalacioBausDevroye:18c,MondelliHassaniUrbanke:18c}).
We develop lower and upper bounds on the random code capacity, and give conditions under which the 
 capacity coincides with the random code capacity, and conditions under which it is lower.
Then, we establish the capacity of the primitive counterpart of the Gaussian AVRC with SFD, in which case  the deterministic and random code capacities coincide.
 
\section{Definitions}
\label{sec:def}
\subsection{Notation}
\label{sec:notation}
We use the following notation conventions throughout. 
Lowercase letters $x,s,y,z,\ldots$  stand for constants and values of random variables, and uppercase letters $X,S,Y,Z,\ldots$ stand for random variables.  
 The distribution of a random variable $X$ is specified by a cumulative distribution function (cdf) 
	$F_X(x)=\prob{X\leq x}$ over the real line $\mathbb{R}$. Alternatively, the distribution may be specified by the probability density function $p(x)$.
 We use $x^k=(x_1,x_{2},\ldots,x_k)$ to denote  a vector in $\mathbb{R}^k$, 
and the short notation $\xvec=(x_1,x_{2},\ldots,x_n)$, 
when it is understood from the context that the length of the vector is $n$.  The $\ell^2$-norm of $\xvec$ is denoted  by $\norm{\xvec}$.  
 A random sequence $\Xvec$ and its distribution $F_{\Xvec}(\xvec)=\prob{X_1\leq x_1,\ldots,X_n\leq x_n}$ are defined accordingly. 
For a pair of integers $i$ and $j$, $1\leq i\leq j$, we define the discrete interval $[i:j]=\{i,i+1,\ldots,j \}$.

\subsection{Coding}
\label{subsec:RYcoding}
We give the definitions of  deterministic and random codes below, 
where 
the term `code'  refers to a deterministic code.  
%
A $(2^{nR},n)$ code for the Gaussian AVRC with SFD
 consists of the following;  a  message set $[1:2^{nR}]$,  where  $2^{nR}$ is assumed to be an integer, 
an encoder $(\fvec',\fvec''):[1:2^{nR}]\rightarrow\mathbb{R}^{2n}$,
a relay encoder $\fvec_1:\mathbb{R}^n\rightarrow\mathbb{R}^n$ where $\enc_{1,i}:\mathbb{R}^{i-1}\rightarrow \mathbb{R}$, 
$i\in [1:n]$, and a  decoding function $\dec: \mathbb{R}^n\rightarrow [1:2^{nR}]$.
The encoder and the relay satisfy the input constraints
$\norm{\fvec'(m)}^2+\norm{\fvec''(m)}^2\leq n\plimit$ and
$\norm{\fvec_1(\yvec_1)}^2\leq n\plimit_1$
for all $m\in [1:2^{nR}]$ and $\yvec_1\in\mathbb{R}^n$.

At time $i\in [1:n]$, given a message $m\in [1:2^{nR}]$,  the encoder transmits $(x_i',x_i'')=(\enc'_i(m),\enc''_i(m))$, and the relay transmits 
$x_{1,i}=\enc_{1,i}(y_{1,1},\ldots,y_{1,i-1})$. The relay codeword is then given by $\xvec_1=\fvec_1(\yvec_1)\triangleq \left( \enc_{1,i}(y_1^{i-1}) \right)_{i=1}^n$.
The decoder receives the output sequence $\yvec$ and  finds an estimate of the message
 $\hm=g(\yvec)$.
We denote the code by $\code=\left(\fvec'(\cdot),\fvec''(\cdot)
,\fvec_1(\cdot),\dec(\cdot) \right)$.
%
 Define the conditional probability of error given of the code $\code$ given a state sequence 
$\svec\in\mathbb{R}^n$, by 
\begin{align}
\cerr(\code)=\frac{1}{2^{nR}}\sum_{m=1}^{2^{nR}} \int_{\Dset(m,\svec)^c} \frac{1}{(2\pi\sigma^2)^{\nicefrac{n}{2}}}e^{-\norm{\zvec}^2/2\sigma^2} \,d\zvec \,,
\end{align}
 where
\begin{align}
\Dset(m,\svec)= \left\{ \zvec\in\mathbb{R}^n : g\big(\,\fvec'(m)+\fvec_1\big( \fvec''(m)+\zvec \big)+\svec \,\big)=m \right\} \,.
\end{align}
  %
We say that $\code$ is a $(2^{nR},n,\eps)$ code for the Gaussian AVRC if it further satisfies 
 $\cerr(\code)\leq \eps$, 
for all $\svec\in\mathbb{R}^n$ with $\norm{\svec}^2\leq
n\Lambda$.
%
A rate $R$ is called achievable if for every $\eps>0$ and sufficiently large $n$, there exists a  $(2^{nR},n,\eps)$ code. The operational capacity $\opC$ is defined as the supremum of 
 achievable rates. 
 We use the term `capacity' referring to this operational meaning, and in some places we call it the deterministic code capacity in order to emphasize that achievability is measured with respect to  deterministic codes.  

Next, we define a random code for which the encoders-decoder
triplet is drawn with shared randomness. 

\begin{definition}
\label{def:RYcorrC} 
A $(2^{nR},n)$ random code 
 consists of a collection of 
$(2^{nR},n)$ codes $\{\code_{j}
\}_{j\in\Jset
}$, along with a probability mass function $\pi$ 
 over the code collection. 
%
For 
  a $(2^{nR},n,\eps)$ random code, 
 $\sum_{j\in\Jset} \pi(j)\cerr(
\code_j)\leq$ $ \eps$, 
for all $\svec\in\mathbb{R}^n$ with $\norm{\svec}^2\leq n\Lambda$. 
The capacity achieved by random codes is denoted by $\opC^{\rstarC}$, and it 
 is referred to as the random code capacity.
\end{definition}
\section{Main Results}
\label{sec:RYres}
Our results are given below.
We determine the random code capacity and  give bounds 
on 
 the deterministic code capacity of the Gaussian AVRC with SFD. 
%
For every $0\leq \alpha,\rho\leq 1$, let
\begin{align}
\mathsf{F}_G(\alpha,\rho)\triangleq 
\min \bigg\{ \frac{1}{2}\log\left( 1+\frac{\plimit_1+\alpha\plimit+2\rho\sqrt{\alpha\plimit\cdot \plimit_1}}{\Lambda} \right) ,\;  
\frac{1}{2}\log\left( 1+\frac{(1-\alpha)\plimit}{\sigma^2} \right)+\frac{1}{2}\log\left( 1+\frac{(1-\rho^2)\alpha\plimit}{\Lambda} \right) 
\bigg\} \,.
\label{eq:Fg}
\end{align}

\begin{theorem}
\label{theo:RYgavcCr}
The random code capacity of the Gaussian AVRC with SFD, under input constraints $\plimit$ and $\plimit_1$ and state constraint $\Lambda$, is given by 
\begin{align}
\label{eq:RYrCavcGauss} 
\opC^{\rstarC} =  
\max_{ 0\leq \alpha,\rho\leq 1} \mathsf{F}_G(\alpha,\rho)
\,.
\end{align}
\end{theorem}
 The proof of Theorem~\ref{theo:RYgavcCr} is given in Appendix~\ref{app:RYgavcCr}, following the considerations in   \cite{PeregSteinberg:17c4}.
%
Next, we give lower and upper bounds on the deterministic code capacity. 
Define
\begin{align}
&\begin{array}{lll}
\inR_{G,low}
\triangleq &\mathrm{max} &\mathsf{F}_G(\alpha,\rho) \\
& \text{$\mathrm{subject}$ $\mathrm{to}$} & 0\leq \alpha,\rho\leq 1\,, \\
&& (1-\rho^2)\alpha\plimit>\Lambda \,,\\
&& \frac{\plimit_1}{\plimit}(\sqrt{\plimit_1}+\rho\sqrt{\alpha\plimit})^2 > \\&&\qquad\qquad \Lambda+(1-\rho^2)\alpha\plimit \,,
\end{array}
\label{eq:RYGlow}
\\
&\begin{array}{lll}
\inR_{G,up}
\triangleq &\mathrm{max} &\mathsf{F}_G(\alpha,\rho) \\
& \text{$\mathrm{subject}$ $\mathrm{to}$} & 0\leq \alpha,\rho\leq 1\,, \\
&& \plimit_1+\alpha\plimit+2\rho\sqrt{\alpha\plimit\cdot\plimit_1} \geq\Lambda 
\end{array}
\label{eq:RYGup}
\end{align}
It can be seen that $\inR_{G,low}
\leq \inR_{G,up}
$,
since
\begin{align}
\plimit_1+\alpha\plimit+2\rho\sqrt{\alpha\plimit\cdot\plimit_1}
=(\sqrt{\plimit_1}+\rho\sqrt{\alpha\plimit})^2+(1-\rho^2)\alpha\plimit
 \geq (1-\rho^2)\alpha\plimit \,.
\end{align}
Observe that if $\plimit_1>\Lambda$, then the random code capacity is given by $\opC^{\rstarC} =\inR_{G,up}$ by Theorem~\ref{theo:RYgavcCr}, as the expressions on the RHS of (\ref{eq:RYrCavcGauss}) and (\ref{eq:RYGup}) coincide.
Furthermore, if $\plimit_1$ is large enough, and $\plimit>\sigma^2>\Lambda$, then 
\begin{align}
\mathsf{F}_G(\alpha,\rho)= 
\frac{1}{2}\log\left( 1+\frac{(1-\alpha)\plimit}{\sigma^2} \right)+\frac{1}{2}\log\left( 1+\frac{(1-\rho^2)\alpha\plimit}{\Lambda} \right) \leq \mathsf{F}_G(1,0)
 \,,
\label{eq:FgLp}
\end{align}
which implies that the bounds coincide, and $
\opC^{\rstarC}=\inR_{G,low} =\inR_{G,up}=\frac{1}{2}\log\left( 1+\frac{\plimit}{\Lambda} \right)$.

The deterministic code analysis is based on the following lemma by \cite{CsiszarNarayan:91p}. 
\begin{lemma}[see {\cite[Lemma 1]{CsiszarNarayan:91p}}]
\label{lemm:CN91}
For every $\eps>0$, $8\sqrt{\eps}<\eta<1$, $K>2\eps$, and  $\dM=2^{nR}$, with
 $2\eps\leq R\leq K$, and $n\geq n_0(\eps,\eta,K)$, there exist $\dM$ unit vectors 
$\avec(m)\in \mathbb{R}^n$, $m\in [1:\dM]$,
 such that for every unit vector $\cvec\in\mathbb{R}^n$ and $0\leq \theta,\zeta\leq 1$, 
\begin{align*}
&\big| \big\{  \tm\in [1:\dM] : \langle \avec(\tm),\cvec \rangle \geq \theta   \big\} \big| \leq
2^{n\left( [R+\frac{1}{2}\log(1-\theta^2)]_{+}+\eps \right) } \,,
\intertext{
and if $\theta\geq \eta$ and $\theta^2+\zeta^2>1+\eta-2^{-2R}$, then
}
&\frac{1}{\dM}\big| \big\{  m\in [1:\dM] :  |\langle \avec(\tm),\avec(m) \rangle|\geq \theta  \,,\; \nonumber\\
&\hspace{1cm}
|\langle \avec(\tm),\cvec \rangle| \geq \zeta
\,,\;  
\text{for some $\tm\neq m$}
  \big\} \big| \leq
2^{-n\eps} \,,
\end{align*}
where  $[t]_{+}=\max\{0,t\}$ and  $\langle \cdot,\cdot\rangle$ denotes  inner product.
\end{lemma}
Intuitively, the lemma states that under certain conditions,  a codebook can be constructed with an exponentially small fraction of ``bad" messages, for which the codewords are non-orthogonal to each other and the state sequence.

\begin{theorem}
\label{theo:RYgavcCdet}
The capacity of the Gaussian AVRC with SFD, 
under input constraints $\plimit$ and $\plimit_1$ and state constraint $\Lambda$, is bounded by 
\begin{align}
\label{eq:RYICavcGaussDet} 
\inR_{G,low}
\leq \opC\leq \inR_{G,up}
 \,.
\end{align}
\end{theorem}
The proof of Theorem~\ref{theo:RYgavcCdet} is given in Appendix~\ref{app:RYgavcCdet}.
Figure~\ref{fig:gaussAVRC} depicts the bounds on the capacity of the Gaussian AVRC with SFD under input and state constraints, as a function of the input constraint $\plimit=\plimit_1$, under state constraint $\Lambda=1$ and
$\sigma^2=0.5$. The top dashed line depicts the random code capacity of the Gaussian AVRC. The solid lines depict the deterministic code lower and upper bounds $\inR_{G,low}
$ and $\inR_{G,up}
$.
For low values, $\plimit<\frac{\Lambda}{4}=0.25$, we have that $\inR_{G,up}
=0$, hence the deterministic code capacity is zero, and it is strictly lower than the random code capacity. 
The dotted lower line depicts the direct transmission lower bound, 
which equals $\mathsf{F}_G(1,0)$ for $\plimit>\Lambda$, and zero otherwise  (see (\ref{eq:Gavc1})). 
 For intermediate values of $\plimit$, direct transmission is better than the lower bound in Theorem~\ref{theo:RYgavcCdet}. Whereas, for high values of $\plimit$, our bounds are tight, and the capacity coincides with the random code capacity, \ie  $\RYCavc=\RYrCav=\inR_{G,low}
=\inR_{G,up}
$.
\begin{center}
\begin{figure}[htb]
        \centering
        \includegraphics[scale=0.35,trim={0 0 0 0.7cm},clip] 
				{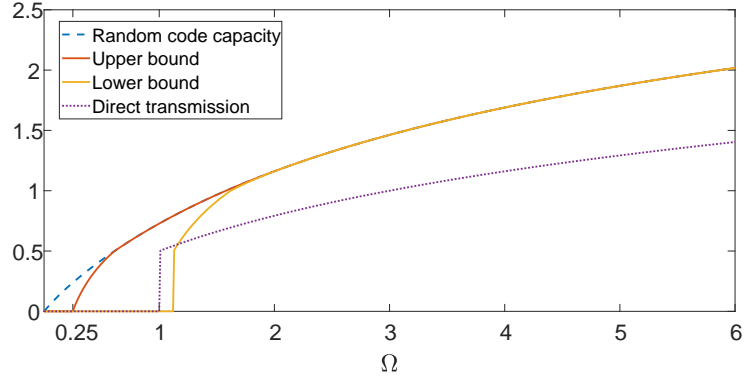}
        
\caption{Bounds on the capacity of the Gaussian AVRC with SFD.
The dashed upper line depicts the random code capacity of the Gaussian AVRC as a function of the input constraint $\plimit=\plimit_1$, under state constraint $\Lambda=1$ and $\sigma^2=0.5$. The solid lines depict the deterministic code lower and upper bounds $\inR_{G,low}
$ and $\inR_{G,up}
$.
The dotted lower line depicts the direct transmission lower bound.
  }
\label{fig:gaussAVRC}
\end{figure}
\end{center}

\section{Primitive Relay Channels}
In this section, we consider the special class of primitive relay channels \cite{Kim:07c}, where the relay communicates with the receiver over a noiseless link of rate $C_1$. 
First, we consider a discrete channel. 
For a discrete random variable $X$ with values in a finite set $\Xset$, we use the notation $\left(p(x)\right)_{x\in\Xset}$ for the probability mass function (pmf). The set of all pmfs on $\Xset$ is denoted by $\pSpace(\Xset)$.	

	\subsection{Channel Description}
	\label{subsec:pRYCs}

A state-dependent discrete memoryless primitive relay channel 
$(\Xset,\Sset,\prc,\Yset,\Yset_1)$ consists of the sets $\Xset$,  $\Sset$, $\Yset$ and $\Yset_1$, and a collection of conditional pmfs $\prc$.  
The sets stand for the input alphabet,  the state alphabet, the output alphabet, and the relay input alphabet, respectively.  The alphabets are assumed to be finite, unless explicitly said otherwise.
 The channel is memoryless without feedback, and therefore  
\begin{align}
W_{Y^n,Y_1^n|X^n,S^n}(y^n,y_1^n|x^n,s^n)= \prod_{i=1}^n \prc(y_{i},y_{1,i}|x_i,s_i) \,.
 \end{align}
Communication over a primitive relay channel is depicted in Figure~\ref{fig:RCprimitive}.
At first, we consider a general channel, not necessarily Gaussian nor with orthogonal sender components as in (\ref{eq:GsfdDef}).
A primitive relay channel is said to be degraded if
\begin{align}
\prc(y,y_1|x,s)=W_{Y_1|X,S}(y_1|x,s)W_{Y|Y_1,S}(y_1|y,s) \,,
\label{eq:pDRC}
\end{align}
and reversely degraded if
\begin{align}
\prc(y,y_1|x,s)=W_{Y|X,S}(y|x,s)W_{Y_1|Y,S}(y_1|y,s) \,.
\label{eq:prevDRC}
\end{align}
We say that the primitive relay channel is \emph{strongly degraded} if (\ref{eq:pDRC}) holds such that the channel from $X$ to $Y_1$ does not depend on the state, \ie $W_{Y_1|X,S}=W_{Y_1|X}$. For example, a primitive
relay channel specified by $Y_1=X+Z$ and $Y=Y_1+S=X+S+Z$ is strongly degraded, for an additive noise $Z$ which is independent of the state.
Similarly, $\prc$ is reversely strongly degraded if (\ref{eq:prevDRC}) holds with $W_{Y_1|Y,S}=W_{Y_1|Y}$.
For example, a primitive
relay channel with $Y=X+S$ and $Y_1=Y+Z=X+S+Z$ is reversely strongly degraded. 

Instances of channels that meet the description in part 1 and part 2 of the corollary above are \eg
$Y_1=Y+Z=X+S+Z$ and $Y=Y_1+S=X+S+Z$, respectively, where $Z$ is an independent additive noise.

The primitive AVRC  $\pavrc=\{\prc\}$ 
 is a discrete memoryless primitive relay channel $(\Xset,\Sset,$ $\prc,\Yset,\Yset_1)$  with a state sequence of unknown distribution,  not necessarily independent nor stationary. That is, $S^n\sim \qn(s^n)$ with an unknown joint pmf $\qn(s^n)$ over $\Sset^n$.
In particular, $\qn(s^n)$ can give mass $1$ to some state sequence $s^n$. 

To analyze the primitive AVRC, we also consider the compound primitive  relay channel $\pRYcompound$, 
governed by a discrete memoryless state $S\sim q(s)$, where $q(s)$ is not known in exact, but rather belongs to a family of distributions $\Qset\subseteq \pSpace(\Sset)$.

\begin{center}
\begin{figure}[htb]
        \centering
        \includegraphics[scale=0.65,trim={1cm 0 0 0},clip] 
				{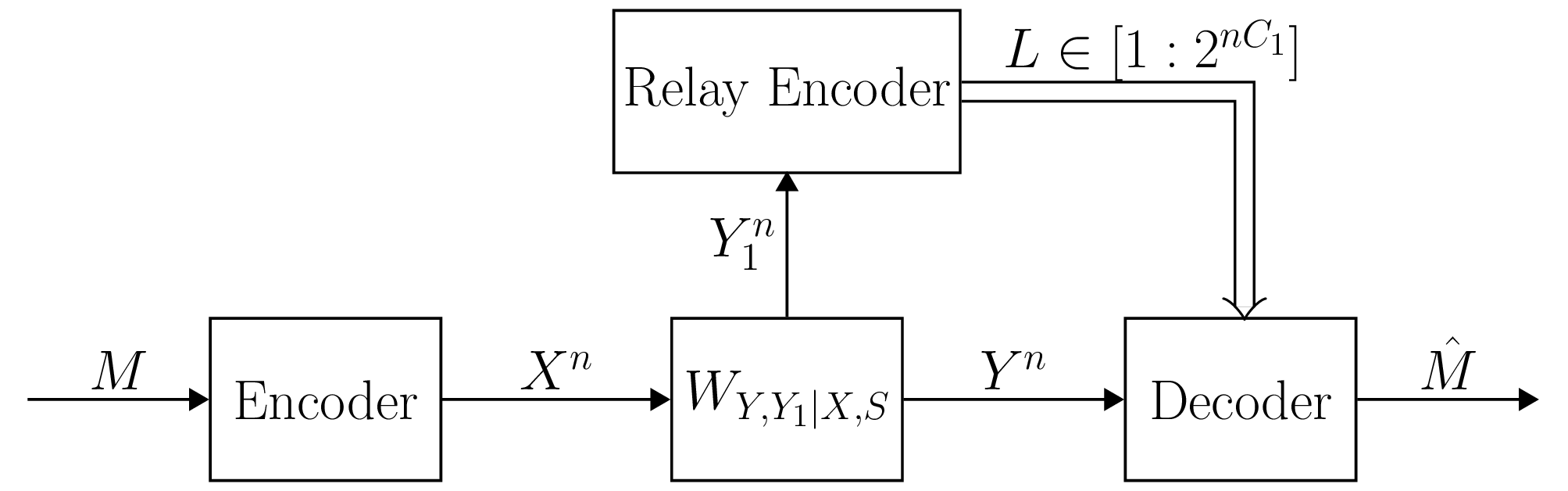}
        
\caption{Communication over the primitive AVRC $\pavrc=\{\prc\}$. 	 Given a message $M$, 
 the encoder transmits $X^n=\enc(M 
)$. 
The relay receives $Y_1^n$ and sends $L=f_1(Y_1^n)$, where $f_1:\Yset_1^n\rightarrow [1:2^{n C_1}]$.
 The decoder receives both the channel output sequence $Y^n$ and the relay output $L$,  and  finds an estimate of the message $\;\hM=g(Y^n,L)$. 
  }
\label{fig:RCprimitive}
\end{figure}
\end{center}

\subsection{Coding}
\label{subsec:pRYcoding}
We introduce some preliminary definitions, starting with the definitions of a deterministic code and a random code for the primitive AVRC $\pavrc$. 

\begin{definition}[A code, an achievable rate and capacity]
\label{def:pRYcapacity}
A $(2^{nR},n)$ code for the primitive AVRC $\pavrc$ 
 consists of the following;  a  message set $[1:2^{nR}]$,  where it is assumed throughout that $2^{nR}$ is an integer, 
an encoder $\enc:[1:2^{nR}]\rightarrow\Xset^n$,
a  relaying function $\enc_{1,i}:\Yset_1^n\rightarrow [1:2^{nC_1}]$, and a  decoding function 
$\dec: \Yset^n\times [1:2^{nC_1}]\rightarrow [1:2^{nR}]$.

 Given a message $m\in [1:2^{nR}]$, 
 the encoder transmits $x^n=\enc(m 
)$. 
The relay receives the sequence $y_1^n$ and sends an index $\ell=\enc_{1}(y_1^n)$.
The decoder receives both the output sequence $y^n$ and the index $\ell$, and  finds an estimate of the message $\hm=g(y^n,\ell)$
(see Figure~\ref{fig:RCprimitive}). 
We denote the code by $\code=\left(\enc(\cdot
),\encn_1(\cdot),\dec(\cdot,\cdot) \right)$.
 Define the conditional probability of error of the code $\code$ given a state sequence $s^n\in\Sset^n$ by  
\begin{align}
\label{eq:pRYcerr}
\cerr(\code)=
\frac{1}{2^{nR}}\sum_{m=1}^{2^ {nR}}
\sum_{ 
\substack{
(y^n,y_1^n)\,:\; \dec(y^n,f_1(y_1^n))\neq m
}
} 
W_{Y^n,Y_1^n|X^n,S^n}(y^n,y_1^n|f(m),s^n) \,.
\end{align}
Now, define the average probability of error of $\code$ for some distribution $\qn(s^n)\in\pSpace(\Sset^n)$, 
\begin{align}
\err(\qn,\code)
=\sum_{s^n\in\Sset^n} \qn(s^n)\cdot\cerr(\code) \,.
\end{align}
Observe that $ \err(\qn,\code)$ is linear  in $q$, and thus continuous. 
We say that $\code$ is a 
$(2^{nR},n,\eps)$ code for the primitive AVRC $\pavrc$ if it further satisfies 
\begin{align}
\label{eq:pRYerr}
 \err(\qn,\code)\leq \eps \,,\quad\text{for all $\qn(s^n)\in\pSpace(\Sset^n)\,$.} 
\end{align} 
A rate $R$ is called achievable if for every $\eps>0$ and sufficiently large $n$, there exists a  $(2^{nR},n,\eps)$ code. The operational capacity is defined as the supremum of the achievable rates and it is denoted by $\pRYCavc$. 
\end{definition} 

We proceed now to define the parallel quantities when using stochastic-encoders stochastic-decoder triplets with common randomness.

\begin{definition}[Random code]
\label{def:pRYcorrC} 
A $(2^{nR},n)$ random code for the primitive AVRC $\pavrc$ consists of a collection of 
$(2^{nR},n)$ codes $\{\code_{\gamma}=(\enc_\gamma,\encn_{1,\gamma},\dec_{\gamma})\}_{\gamma\in\Gamma}$, along with a probability distribution $\mu(\gamma)$ over the code collection $\Gamma$. 
We denote such a code by $\gcode=(\mu,\Gamma,\{\code_{\gamma}\}_{\gamma\in\Gamma})$.
Analogously to the deterministic case,  a $(2^{nR},n,\eps)$ random code has the additional requirement
\begin{align}
 \err(\qn,\gcode)=\sum_{\gamma\in\Gamma} \mu(\gamma)\err(q,
\code_\gamma)\leq \eps \,,\;\text{for all $\qn(s^n)\in\pSpace(\Sset^n)$} & \,. \qquad
\end{align}
The random code capacity 
is denoted by $\pRYrCav$. 
\end{definition}

\section{Main Results -- Primitive AVRC}
\label{sec:pRYres}
We present our results on the compound primitive relay channel and the primitive AVRC. 

\subsection{The Compound Primitive Relay Channel} 
\label{sec:pRYcompound}
 We establish the cutset upper bound and the partial decode-forward lower bound for the compound primitive relay channel. 
 Consider a given compound primitive relay channel $\pRYcompound$. 
Let 
\begin{align}  
\label{eq:pRYICcompound} 
\pRYICcompound \triangleq& \inf_{q\in\Qset}  \max_{p(x)} 
\min \left\{ I_q(X;Y)+C_1 \,,\; I_q(X;Y,Y_1)
\right\} \,,
\intertext{and}
\pRYdIRcompound\triangleq&
\max_{p(u,x)} 
\min \Big\{ \inf_{q\in\Qset} I_q(U;Y)+\inf_{q\in\Qset} I_q(X;Y|U)+ C_1 \,,\; \nonumber\\&\qquad\qquad\quad
\inf_{q\in\Qset} I_{q}(U;Y_1)+ \inf_{q\in\Qset} I_q(X;Y|U)
\Big\} \,,
\label{eq:pRYdIRcompound}
\end{align}
where the subscripts `$CS$' and `$DF$' stand for `cutset' and `decode-forward', respectively.
 
\begin{lemma}
\label{lemm:pRYcompoundDF}
The capacity of the compound primitive relay channel $\pRYcompound$ is bounded by
\begin{align}
 &\pRYCcompound \geq \pRYdIRcompound \,, \\
 &\pRYrCcompound \leq \pRYICcompound \,.
\end{align}
Specifically, if $R<\pRYdIRcompound$, then there exists a $(2^{nR},n,e^{-an})$ block Markov code over $\pRYcompound$ for 
 sufficiently large $n$ and some $a>0$.
\end{lemma}
The proof of Lemma~\ref{lemm:pRYcompoundDF} is given in Appendix~\ref{app:pRYcompoundDF}.
Observe that taking $U=\emptyset$  in (\ref{eq:pRYdIRcompound}) gives the direct transmission lower bound,
\begin{align}
\pRYCcompound\geq& \pRYdIRcompound \geq \max_{p(x)} \inf_{q\in\Qset} I_q(X;Y) \,.
\label{eq:pRYcompoundDirectTran}
\intertext{
In particular, this is the capacity when $C_1=0$, \ie when there is no relay.
Taking $U=X$  in (\ref{eq:pRYdIRcompound}) results in a full decode-forward lower bound,
}
\pRYCcompound\geq& \pRYdIRcompound \geq \max_{p(x)} \inf_{q\in\Qset}  \min \left\{ I_q(X;Y)+C_1 \,,\; 
 I_{q}(X;Y_1)
\right\} \,.
\label{eq:pRYcompoundFullDF}
\end{align}
This yields the following corollary.
\begin{coro}
\label{coro:pRYcompoundDeg}
Let $\pRYcompound$ be a compound primitive relay channel, where $\Qset$ is a compact convex set.
\begin{enumerate}[1)]
\item
If $\prc$ is reversely strongly degraded, 
 then
\begin{align}
\pRYCcompound= \pRYdIRcompound=\pRYICcompound= \min_{q\in\Qset} \max_{p(x)}  I_q(X;Y) \,.
\end{align}
\item
If $\prc$ is strongly degraded, 
 then 
\begin{align}
\pRYCcompound= \pRYdIRcompound=\pRYICcompound=  \max_{p(x)}   \min \left\{ \min_{q\in\Qset} I_q(X;Y)+C_1 \,,\; 
 I(X;Y_1)
\right\} \,.
\end{align}
\end{enumerate}
\end{coro}
The proof of Corollary~\ref{coro:pRYcompoundDeg} is given in Appendix~\ref{app:pRYcompoundDeg}. 
Intuitively, in the case of a reveresly strongly degraded relay channel, the relay is useless, while in the case of a strongly degraded relay channel, the relay could be valuable and does not creat a bottleneck. Indeed, 
part 1 follows from the direct transmission and cutset bounds, (\ref{eq:pRYcompoundDirectTran}) and  (\ref{eq:pRYICcompound}), respectively, while  part 2 is based on the full decode-forward and cutset bounds, (\ref{eq:pRYcompoundFullDF}) and  (\ref{eq:pRYICcompound}), respectively.

\subsection{The Primitive AVRC}
We give lower and upper bounds, on the random code capacity and the deterministic code capacity, for the 
primitive AVRC $\pavrc$. 
\subsubsection{Random Code Lower and Upper Bounds}
Define 
\begin{align}
\label{eq:pRYIRcompoundP}  
\pRYdIRavc \triangleq \pRYdIRcompound\bigg|_{\Qset=\pSpace(\Sset)} 
\,,\;
\pRYrICav\triangleq \pRYICcompound\bigg|_{\Qset=\pSpace(\Sset)}
\,.
\end{align}

\begin{theorem}
\label{theo:pRYmain}
The random code capacity  of a primitive AVRC $\pavrc$ is bounded by
\begin{align}
\pRYdIRavc\leq \pRYrCav \leq \pRYrICav \,.
\end{align}
\end{theorem}
The proof of Theorem~\ref{theo:pRYmain} is given in Appendix~\ref{app:pRYmain}. 
Together with Corollary~\ref{coro:pRYcompoundDeg}, this yields another corollary.
\begin{coro}
\label{coro:pRYavrcDeg}
Let $\pavrc$ be a primitive AVRC.
\begin{enumerate}[1)]
\item
If $\prc$ is reversely strongly degraded, 
 then
\begin{align}
\pRYrCav= 
  \min_{q(s)} \max_{p(x)} I_q(X;Y) \,.
\end{align}
\item
If $\rc$ is  strongly degraded, 
 then
\begin{align}
\pRYrCav= 
  \max_{p(x)}   \min \left\{ \min_{q(s)} I_q(X;Y)+C_1 \,,\; 
 I(X;Y_1)
\right\} \,.
\end{align}
\end{enumerate}
\end{coro}

Before we proceed to the deterministic code capacity, we 
note that Ahlswede's Elimination Technique \cite{Ahlswede:78p} 
applies to the primitive AVRC as well. Hence, the size of the code collection of any reliable random code can be reduced to polynomial size.

\subsubsection{Deterministic Code Lower and Upper Bounds}
In the next statements, we characterize the deterministic code capacity of the primitive AVRC $\avrc$. 
We consider conditions under which the deterministic code capacity coincides with the random code capacity, and conditions under which it is lower.
Denote 
 the marginal AVC from the sender to the relay by 
 \begin{align}
\label{eq:pRYmarginAVCs}
\avc_1=\{ W_{Y_1|X,S}\}
 \,, 
\end{align}
and denote the corresponding capacity by $\opC(\avc_1)$. 
\begin{lemma}
\label{lemm:pRYcorrTOdetC} 
If $\min\left(C_1,\opC(\avc_1)\right)>0$,  then the capacity of the primitive AVRC $\pavrc$ coincides with the random code capacity, \ie 
$\pRYCavc = \pRYrCav$.  
\end{lemma}
The proof of Lemma~\ref{lemm:pRYcorrTOdetC} is given in Appendix~\ref{app:pRYcorrTOdetC}, using Ahlswede's Elimination Technique \cite{Ahlswede:78p}. 
Based on the results by \cite{CsiszarNarayan:88p} for the single-user AVC, we give a computable sufficient condition, under which the deterministic code capacity coincides with the random code capacity.
 The condition is given in terms of channel symmetrizability, the definition of which is given below.
\begin{definition} 
\label{def:symmetrizable} \cite{Ericson:85p,CsiszarNarayan:88p}
 A state-dependent DMC $\channel$ is said to be \emph{symmetrizable} if for some conditional distribution $J(s|x)$,
\begin{align}
\label{eq:symmetrizable}
\sum_{s\in\Sset} \channel(y|x,s)J(s|\tx)= \sum_{s\in\Sset} \channel(y|\tx,s)J(s|x) \,,
\end{align}
for all $x,\tx\in\Xset$, $y\in\Yset$.
Equivalently, the channel $\widetilde{W}(y|x,\tx)$ $=$ $
\sum_{s\in\Sset} \channel(y|x,s)J(s|\tx)$ is symmetric, \ie $\widetilde{W}(y|x,\tx)=\widetilde{W}(y|\tx,x)$, for all $x,\tx\in\Xset$, $y\in\Yset$. 
\end{definition}
Intuitively, symmetrizability identifies a poor channel, where 
the jammer can impinge the communication scheme by randomizing the state sequence $S^n$ according to $J^n(s^n|\tx^n)=\prod_{i=1}^n J(s_i|\tx_{i})$, for some codeword $\tx^n$. 
 While the transmitted codeword is $x^n$, the codeword $\tx^n$ can be thought of as an impostor sent by the jammer.  Now, since  the ``average channel" $\widetilde{W}$ is symmetric with respect to $x^n$ and  $\tx^n$, the two codewords appear to the receiver as equally likely. Indeed, by \cite{Ericson:85p}, if an AVC $\{\channel\}$ without a relay  is symmetrizable, then its capacity is zero. 
Furthermore,  Csisz{\'{a}}r and Narayan \cite{CsiszarNarayan:88p} proved that non-symmetrizability is
not only a necessary condition for a positive capacity, but it is a sufficient condition as well. 

\begin{theorem}
\label{theo:pRYmainDbound}
Let $\pavrc$ be a primitive AVRC.
\begin{enumerate}[1)]
\item
If $W_{Y_1|X,S}$ is non-symmetrizable and $C_1>0$, then  
$\pRYCavc = \pRYrCav$. In this case, 
$
\pRYdIRavc\leq \pRYCavc \leq \pRYrICav 
$. 
\item
If $\prc$ is reversely strongly degraded, 
where $W_{Y_1|X,S}$ is non-symmetrizable and $C_1>0$,
 then
\begin{align}
\pRYCavc=
  \min_{q(s)} \max_{p(x)} I_q(X;Y) \,.
\end{align}
\item
If $\prc$ is strongly degraded, 
such that $W_{Y_1|X}(y_1|x)\neq W_{Y_1|X}(y_1|\tx)$ for some $x,\tx\in\Xset$, $y_1\in\Yset_1$, and $C_1>0$,
 then
\begin{align}
\pRYCavc=
  \max_{p(x)}   \min \left\{ \min_{q(s)} I_q(X;Y)+C_1 \,,\; 
 I(X;Y_1)
\right\} \,.
\end{align}
\item
If $W_{\tY|X,S}$ is symmetrizable, where $\tY=(Y,Y_1)$, then $\pRYCavc=0$.
\end{enumerate}
\end{theorem}
 The proof of Theorem~\ref{theo:pRYmainDbound} is given in Appendix~\ref{app:pRYmainDbound}.

To illustrate our results, we give the following example of a primitive AVRC.
\begin{example}
\label{example:pBSRCsymm}
Consider a state-dependent primitive relay channel  $\prc$, specified by 
\begin{align}
Y_1=& X(1-S)  \,, \nonumber\\
Y=& X+S						 \,, \nonumber
\end{align}
where $\Xset=\Sset=\Yset_1=\{ 0,1 \}$, $\Yset=\{0,1,2\}$, and $C_1=1$, \ie the link between the relay and the receiver is a noiseless bit pipe. 
It can be seen that both the sender-relay and the sender-receiver marginals are symmetrizable. Indeed,
$W_{Y|X,S}$ satisfies (\ref{eq:symmetrizable}) with $J(s|x)=1$ for $s=x$, and $J(s|x)=0$ otherwise, 
while $W_{Y_1|X,S}$ satisfies (\ref{eq:symmetrizable}) with $J(s|x)=1$ for $s=1-x$, and $J(s|x)=0$ otherwise.
Nevertheless, the capacity of the primitive AVRC $\pavrc=\{\prc\}$ is 
$\pRYCavc=1$, which can be achieved using a code of length $n=1$, with
$f(m)=m$, $f_1(y_1)=y_1$,
\begin{align}
& g(y,\ell)=g(y,y_1)=\begin{cases}
0 & y=0\\
1 & y=2\\
y_1 & y=1\\
\end{cases}
\end{align}
for $m,y_1\in\{0,1\}$ and $y\in\{0,1,2\}$.
%
This example shows that even if the sender-relay and sender-receiver marginals are symmetrizable, the capacity 
 may still be positive. We further note that
  the condition in part 4 of Theorem~\ref{theo:pRYmainDbound} implies that $W_{Y|X,S}$ and $W_{Y_1|X,S}$ are both symmetrizable, but not vice versa, as shown by this example. 
\end{example}

\subsection{Primitive Gaussian AVRC}
Consider the primitive Gaussian 
 relay channel 
 with SFD, 
\begin{align}
Y_1 =& X''+Z \,,\nonumber\\
Y=& X'+S \,,
\end{align}
Suppose that $C_1>0$, and input and state constraints are imposed as before, \ie
 $\frac{1}{n}\sum_{i=1}^n (X_i'^2+X_i''^2)\leq\plimit$ and $\frac{1}{n}\sum_{i=1}^n S_i^2\leq \Lambda$ with probability $1$. 
The capacity of the primitive Gaussian AVRC with SFD, under input constraint $\plimit$  and state constraint $\Lambda$ is given by 
\begin{align}  
\label{eq:pRYrCavcGauss} 
\pRYCavc=\pRYrCav =  
\max_{ 0\leq \alpha\leq 1} \bigg[
 \frac{1}{2}\log\left( 1+\frac{\alpha\plimit}{\Lambda} \right) 
  + \min\left\{ C_1, \frac{1}{2}\log\left( 1+\frac{(1-\alpha)\plimit}{\Lambda} \right) \right\}    \bigg]
\,.
\end{align}
This result is due to the following. Observe that one could treat this primitive AVRC as two independent channels, 
one from $X'$ to $Y$ and the other from $X''$ to $Y_1$, dividing the input power to $\alpha\plimit$ and 
$(1-\alpha)\plimit$, respectively. 
Based on this observation, the random code direct part  follows from \cite{HughesNarayan:87p}. Next, the deterministic code direct part follows from part 1 of Theorem~\ref{theo:pRYmainDbound}, and the converse part follows straightforwardly from the cutset upper bound in Theorem~\ref{theo:pRYmain}. 

\begin{appendices}

\section{Proof of Theorem~\ref{theo:RYgavcCr}}
\label{app:RYgavcCr}
Consider the Gaussian AVRC with SFD under input constraints $\plimit$ and $\plimit_1$ and state constraint $\Lambda$. We prove the theorem using the partial decode-forward lower bound and the cutset upper bound in \cite{PeregSteinberg:17c4}.

\subsection{Achievability Proof}
We begin with the following lemma, which follows from \cite{HughesNarayan:87p} and \cite{CsiszarNarayan:88p1}.
\begin{lemma}[see \cite{HughesNarayan:87p,CsiszarNarayan:88p1}]
\label{lemm:GaussMinI}
Let $\bar{X}$ be a Gaussian random variable  with variance $P$.
Then, for every $\bar{S}\sim q(\bar{s})$ with $\var(\bar{S})\leq N$, $N>0$,
\begin{align}
 I_q(\bar{X};\bar{X}+\bar{S}) \geq \frac{1}{2} \log\left( 1+\frac{P}{N} \right)\,,
\end{align}
with equality for $\bar{S}\sim\mathcal{N}(0,N)$.
\end{lemma}
\begin{proof}[Proof of Lemma~\ref{lemm:GaussMinI}]
Consider the additive-state AVC, 
specified by $\bar{Y}=\bar{X}+\bar{S}$, under input constraint $P$ and state constraint $N$.
Then, by Csisz{\'{a}}r and Narayan  \cite{CsiszarNarayan:88p1}, the random code capacity of the AVC $\avc$, under input constraint $P$ and state constraint $N$, is given by
\begin{align}
\rCav=\min_{q(\bar{s}) \,:\; \E \bar{S}^2\leq N} \max_{p(\bar{x}) : \E X^2\leq P} I_q(\bar{X};\bar{Y})
= \max_{p(\bar{x}) \,:\; \E X^2\leq P} \min_{q(\bar{s}) : \E \bar{S}^2\leq N} I_q(\bar{X};\bar{Y}) \,. 
\end{align}
On the other hand, by  Hughes and Narayan \cite{HughesNarayan:87p}, 
\begin{align}
\rCav=\frac{1}{2} \log\left( 1+\frac{P}{N} \right) \,. 
\end{align}
As the saddle point value $I_q(\bar{X},\bar{Y})=\frac{1}{2} \log\left( 1+\frac{P}{N} \right)$ is attained with $\bar{X}\sim\mathcal{N}(0,P)$ and $\bar{S}\sim\mathcal{N}(0,N)$, we have that $\bar{S}\sim\mathcal{N}(0,N)$ minimizes $I_q(\bar{X};\bar{Y})$ for $\bar{X}\sim\mathcal{N}(0,P)$.
\end{proof}
Next, we use the lemma above to prove the direct part.
 Although it is assumed in \cite{PeregSteinberg:17c4} that the input, state and output alphabets are finite, the results can be extended to the continuous case as well, using standard discretization techniques \cite{BBT:59p,Ahlswede:78p} \cite[Section 3.4.1]{ElGamalKim:11b}. In particular, \cite[Lemma 1]{PeregSteinberg:17c4}  can be extended to a relay channel under  input constraints $\plimit$ and $\plimit_1$ and state constraint $\Lambda$, by choosing a distribution $p(x',x'',x_1)$ such that 
$\E (X'^2+X''^2)\leq \plimit$ and $\E X_1^2\leq\plimit_1$.
Thus, the random code capacity of the Gaussian AVRC with SFD is lower bounded by
\begin{align}
 \opC^{\rstarC} \geq  
\max_{\substack{ p(x'') p(x',x_1)\,:\; \\  E(X'^2+X''^2)\leq \plimit \,,\; \\  \E X_1^2\leq\plimit_1} }  
\min \Big\{&  \min_{q(s) \,:\; \E S^2\leq\Lambda} I_q(X_1;Y)+ 
 \min_{q(s) \,:\; \E S^2\leq\Lambda} I_q(X';Y|X_1) \,,\; \nonumber\\&\qquad 
  I(X'';Y_1)+ \min_{q(s) \,:\; \E S^2\leq\Lambda} I_q(X';Y|X_1)
\Big\} \,,
\label{eq:RYdIRcompoundG}
\end{align}
which 
 follows from the partial decode-forward lower bound in \cite[Theorem 3]{PeregSteinberg:17c4} by taking $U=X''$. 
Let $0\leq \alpha,\rho\leq 1$, and let $(X',X'',X_1)$ be jointly Gaussian with
\begin{align}
&X'\sim\mathcal{N}(0,\alpha\plimit) \,,\; X''\sim\mathcal{N}(0,(1-\alpha)\plimit) \,,\; X_1\sim\mathcal{N}(0,\plimit_1) \,,
\end{align}
where the correlation coefficient of $X'$ and $X_1$ is $\rho$, while $X''$ is independent of $(X',X_1)$. Hence, 
\begin{align}
I(X'';Y_1)=\frac{1}{2}\log\left( 1+\frac{(1-\alpha)\plimit}{\sigma^2} \right) \,.
\label{eq:RYgI1}
\end{align}
 By Lemma~\ref{lemm:GaussMinI}, as $\var(X'|X_1=x_1)=(1-\rho^2)\alpha\plimit$ for all $x_1\in\mathbb{R}$, we have that
\begin{align}
\min_{q(s) \,:\; \E S^2\leq\Lambda} I_q(X';Y|X_1)=
\frac{1}{2}\log\left( 1+\frac{(1-\rho^2)\alpha\plimit}{\Lambda} \right) \,.
\label{eq:RYgI2}
\end{align} 
It is left for us to evaluate the first term in the RHS of (\ref{eq:RYdIRcompoundG}).
Then, by standard whitening transformation,
 there exist two independent Gaussian random variables $T_1$ and $T_2$ such that
\begin{align}
&X'+X_1=T_1+T_2 \,,
\\ 
&T_1\sim\mathcal{N}(0,(1-\rho^2)\alpha\plimit) \,,\;
 T_2\sim\mathcal{N}(0,\plimit_1+\rho^2\alpha\plimit+2\rho\sqrt{\alpha\plimit\cdot\plimit_1} )  \,.
\end{align}
Hence, $Y=T_1+T_2+S$, and as $\var(X'|X_1=x_1)=\var(T_1)$ for all $x_1\in\mathbb{R}$, we have that
\begin{align}
I_q(X_1;Y)=& H_q(Y)-H_q(X'+S|X_1) \nonumber\\
=& H_q(Y)-H_q(T_1+S)=I_q(T_2;Y) 
\,
\end{align}
Let $\bar{S}\triangleq T_1+S$. 
Then, by Lemma~\ref{lemm:GaussMinI},
\begin{align}
\min_{q(s) \,:\; \E S^2\leq\Lambda} I_q(X_1;Y)=&\min_{q(s) \,:\; \E S^2\leq\Lambda} I_q(T_2;T_2+\bar{S})
\nonumber\\
=&\frac{1}{2}\log\left(  1+\frac{\var(T_2)}{\var(T_1)+\Lambda} \right) 
=\frac{1}{2}\log\left(  \frac{\plimit_1+\rho^2\alpha\plimit+2\rho\sqrt{\alpha\plimit\cdot\plimit_1}+\Lambda}
{(1-\rho^2)\alpha\plimit+\Lambda} \right) \,.
\label{eq:RYgI3}
\end{align}
Substituting (\ref{eq:RYgI1}), (\ref{eq:RYgI2}) and (\ref{eq:RYgI3}) in the RHS of (\ref{eq:RYdIRcompoundG}), we have that
\begin{align}
 \opC^{\rstarC} \geq&   \max_{ 0\leq \alpha,\rho\leq 1 } 
\min \bigg\{  \frac{1}{2}\log\left(  \frac{\plimit_1+\rho^2\alpha\plimit+2\rho\sqrt{\alpha\plimit\cdot\plimit_1}+\Lambda}
{(1-\rho^2)\alpha\plimit+\Lambda} \right)+ 
 \frac{1}{2}\log\left(1+ \frac{(1-\rho^2)\alpha\plimit}{\Lambda} \right) \,,\; \nonumber\\&
  \frac{1}{2}\log\left( 1+\frac{(1-\alpha)\plimit}{\sigma^2} \right)+ \frac{1}{2}\log\left( 1+\frac{(1-\rho^2)\alpha\plimit}{\Lambda} \right)
\bigg\}
 \,.
\label{eq:RYdIRcompoundG1}
\end{align}
Observe that the first sum in the RHS of (\ref{eq:RYdIRcompoundG1}) can be expressed as
\begin{align}
&\frac{1}{2}\log\left(  \frac{\plimit_1+\rho^2\alpha\plimit+2\rho\sqrt{\alpha\plimit\cdot\plimit_1}+\Lambda}
{(1-\rho^2)\alpha\plimit+\Lambda} \right)+ 
 \frac{1}{2}\log\left( \frac{(1-\rho^2)\alpha\plimit+\Lambda}{\Lambda} \right)
\nonumber\\
=&\frac{1}{2}\log\left( \frac{\plimit_1+\rho^2\alpha\plimit+2\rho\sqrt{\alpha\plimit\cdot\plimit_1}+\Lambda}
{\Lambda}  \right)=\frac{1}{2}\log\left(1+ \frac{\plimit_1+\rho^2\alpha\plimit+2\rho\sqrt{\alpha\plimit\cdot\plimit_1}}
{\Lambda}  \right) \,.
\end{align}
Hence, the direct part follows from (\ref{eq:RYdIRcompoundG1}). \qed

\subsection{Converse Proof}
By \cite[Theorem 3]{PeregSteinberg:17c4}, the random code capacity is upper bounded by
\begin{align}
  \opC^{\rstarC} 
	\leq&  
 \min_{q(s) \,:\; \E S^2\leq \Lambda}  
\max_{\substack{ p(x'') p(x,x_1)\,:\; \\  E(X'^2+X''^2)\leq \plimit \,,\; \\  \E X_1^2\leq\plimit} } 
\min \left\{ I_q(X',X_1;Y) \,,\; I(X'';Y_1)+I_q(X';Y|X_1)
\right\} \nonumber\\
\leq& \max_{\substack{ p(x'') p(x,x_1)\,:\; \\  E(X'^2+X''^2)\leq \plimit \,,\; \\  \E X_1^2\leq\plimit} } 
\min \left\{ I_q(X',X_1;Y) \,,\; I(X'';Y_1)+I_q(X';Y|X_1)
\right\} \bigg|_{S\sim \mathcal{N}(0,\Lambda)} \nonumber\\
=&\max_{ 0\leq \alpha,\rho\leq 1 } 
\min \bigg\{ \frac{1}{2}\log\left( 1+\frac{\plimit_1+\rho^2\alpha\plimit+2\rho\sqrt{\alpha\plimit\cdot\plimit_1}}{\Lambda} \right) 
\,,\; 
\frac{1}{2}\log\left( 1+\frac{(1-\alpha)\plimit}{\sigma^2} \right)+\frac{1}{2}\log\left( 1+\frac{(1-\rho^2)\alpha\plimit}{\Lambda} \right) 
\bigg\}
 \,,
\label{eq:RYICcompoundGconverse} 
\end{align}
where the last equality is due to \cite{ElGamalZahedi:05p}. 
\qed

\begin{center}
\begin{figure}
\begin{tabular}{l|ccccc}
Block				& $1$							& $2$									
									& $\cdots$			& $B-1$ 		& $B$\\
								
\\ \hline &&&&& \\ 
Encoder				&	$\xvec(m_1',m_1''|1)$	&	$\xvec(m_2',m_2''|m_1')$		
							& $\cdots$			& $\xvec(m_{B-1}',m_{B-1}''|m_{B-2}')$&  $\xvec(1,1|m_{B-1}')$ 
\\&&&&& \\
Relay Decoder			& $\tm_1'\rightarrow$				& $\tm_2'\rightarrow$						
						& $\cdots$			& $\tm_{B-1}'$ & $\emptyset$
\\&&&&& \\
Relay Encoder				&	$\xvec_1(1)$	&	$\xvec_1(\tm_1')$		
							& $\cdots$			& $\xvec_1(\tm_{B-2}')$&  $\xvec_1(m_{B-1}')$ \\&&&&& \\
Output 		& $\emptyset$		& $\hm_1'$ 
					& $\cdots$ & $\leftarrow\hm_{B-2}'$ & $\leftarrow\hm_{B-1}'$ 
					\\
					& $\hm_1''$		& $\hm_2''$ 
					& $\cdots$ & $\hm_{B-1}''$ & $\emptyset$
\end{tabular}
\caption{Partial decode-forward coding scheme. The block index $b\in [1:B]$ is indicated at the top. In the following rows, we have the corresponding elements: 
(1) sequences transmitted by the encoder;  
(2) estimated messages at the relay;
(3) sequences transmitted by the relay;
(4) estimated messages at the destination decoder.
The arrows in the second row indicate that the relay encodes forwards with respect to the block index, while the arrows in the fourth row indicate that the receiver decodes backwards.  
}
\label{fig:RYpDFcompound}
\end{figure}
\end{center}

\section{Proof of Theorem~\ref{theo:RYgavcCdet}}
\label{app:RYgavcCdet}
Consider the Gaussian AVRC $\avrc$ with SFD under input constraints $\plimit$ and $\plimit_1$ and state constraint $\Lambda$.

\subsection{Lower Bound}
We construct a block Markov code using backward minimum-distance decoding in two steps.
The encoders use $B$ blocks, each consists of $n$ channel uses, to convey $(B-1)$ independent messages to the receiver, where each message $M_b$, for $b\in [1:B-1]$,  is divided into two independent messages.
That is, $M_b=(M_b',M_b'')$,  where $M_b'$ and $M_b''$ are uniformly distributed, \ie
\begin{align}
M_b'\sim \text{Unif}[1:2^{nR'}] \,,\; M_b''\sim \text{Unif}[1:2^{nR''}] \,,\;\text{with $R'+R''=R$}\,,
\end{align}
for $b\in [1:B-1]$.
For convenience of notation, set $M_0'=M_B'\equiv 1$ and $M_0''=M_B''\equiv 1$. The average rate $\frac{B-1}{B}\cdot R$ is arbitrarily close to $R$. The block Markov coding scheme is illustrated in Figure~\ref{fig:RYpDFcompound}.
%

\emph{Codebook Construction}: 
Fix $0\leq \alpha,\rho\leq 1$ with  
\begin{align}
 &(1-\rho^2)\alpha\plimit>\Lambda \,, \label{eq:RYGassump1} \\
 &\frac{\plimit_1}{\plimit}(\sqrt{\plimit_1}+\rho\sqrt{\alpha\plimit})^2>\Lambda+(1-\rho^2)\alpha\plimit \label{eq:RYGassump2} \,.
\end{align}
We  construct $B$ codebooks $\Fset_b$ of the following form,
\begin{align}
\Fset_b=\Big\{ \left(  \xvec_1(m_{b-1}'),  \xvec'(m_b',m_b''|m_{b-1}'), \xvec''(m_b') \right)   \,:\;
 m_{b-1}',m_b'\in [1:2^{nR'}] \,,\; m_b''\in [1:2^{nR''}]
\Big\} \,,
\end{align}
for $b\in [2:B-1]$. 
 The codebooks $\Fset_1$ and $\Fset_B$ have the same form, with fixed $m_0'=m_B'\equiv 1$ and $m_0''=m_B''\equiv 1$. 

The sequences $\xvec''(m_b')$, $m_b'\in [1:2^{nR'}]$ are chosen as follows. Observe that the channel from the sender to the relay, $Y_1=X''+Z$, does not depend on the state.
Thus,  by Shannon's well-known result on the point to point Gaussian channel \cite{Shannon:48p},  the message $m_b'$ can be conveyed to the relay reliably, under input constraint $(1-\alpha)\plimit$, provided that 
$R'<\frac{1}{2}\log\left(1+\frac{(1-\alpha)\plimit}{\sigma^2}  \right)-\delta_1$, where $\delta_1$ is arbitrarily small (see also \cite[Chapter 9]{CoverThomas:06b}). 
 That is, for every $\eps>0$ and sufficiently large $n$, there exists a $(2^{n R'},n,\eps)$ code 
$\code''=(\xvec''(m_b'),g_1(\yvec_{1,b}))$, such that
$\norm{\xvec''(m_b')}^2\leq n(1-\alpha)\plimit$ for all $m_b'\in [1:2^{nR'}]$.

Next, we choose the sequences $\xvec_1(m_{b-1}')$ and
 $\xvec'(m_b',m_b''|$ $m_{b-1}')$, for $m_{b-1}',$ $m_b'\in [1:2^{nR'}]$, $m_b''\in [1:2^{nR''}]$.
Applying Lemma~\ref{lemm:CN91} by \cite{CsiszarNarayan:91p} repeatedly yields the following.
\begin{lemma}
\label{lemm:RYgKey}
For every $\eps>0$, $8\sqrt{\eps}<\eta<1$, $K>2\eps$, 
$2\eps\leq R'\leq K$, $2\eps\leq R''\leq K$, and
$n\geq n_0(\eps,\eta,K)$,
\begin{enumerate}[1)]
\item
 there exist $2^{nR'}$ unit vectors, 
\begin{align}
&\avec(m_{b-1}')\in \mathbb{R}^n \,,\; m_{b-1}'\in [1:2^{nR'}] \,, 
\end{align}
 such that for every unit vector $\cvec\in\mathbb{R}^n$ and $0\leq \theta,\zeta\leq 1$, 
\begin{align}
&\left| \left\{  \tm_{b-1}'\in [1:2^{nR'}] \,:\;\, \langle \avec(\tm_{b-1}'),\cvec \rangle \geq \theta   \right\} \right| \leq
2^{n\left([R'+\frac{1}{2}\log(1-\theta^2)]_{+}+\eps\right)} \,,
\intertext{
and if $\theta\geq \eta$ and $\theta^2+\zeta^2>1+\eta-2^{-2R'}$, then
}
&\frac{1}{2^{nR'}}\big| \big\{  m_{b-1}'\in [1:2^{nR'}] \,:\;\,  |\langle \avec(\tm_{b-1}'),\avec(m_{b-1}') \rangle|\geq \theta \,,\;
|\langle \avec(\tm_{b-1}'),\cvec \rangle| \geq \zeta \,,\; 
\text{for some $\tm_{b-1}'\neq m_{b-1}'$}
  \big\} \big| \leq
2^{-n\eps} \,.
\end{align}
\item
Furthermore, for every $m_b'\in [1:2^{nR'}]$, there exist $2^{nR''}$ unit vectors, 
\begin{align}
&\vvec(m_b',m_b'')\in \mathbb{R}^n \,,\; m_b''\in  [1:2^{nR''}] \,,
\end{align}
 such that for every unit vector $\cvec\in\mathbb{R}^n$ and $0\leq \theta,\zeta\leq 1$, 
\begin{align}
&\left| \left\{  \tm_{b}''\in  [1:2^{nR''}] \,:\;\, \langle \vvec(m_b',\tm_{b}''),\cvec \rangle \geq \theta   \right\} \right| \leq
2^{n\left([R''+\frac{1}{2}\log(1-\theta^2)]_{+}+\eps \right)} \,,
\intertext{
and if $\theta\geq \eta$ and $\theta^2+\zeta^2>1+\eta-2^{-2R''}$, then
}
&\frac{1}{2^{nR''}}\big| \big\{  m_{b}''\in  [1:2^{nR''}] \,:\;\,  |\langle \vvec(m_b',\tm_{b}''),\vvec(m_{b}',m_b'') \rangle|\geq \theta  \,,\;
|\langle \vvec(m_b',\tm_{b}''),\cvec \rangle| \geq \zeta  \,,\; 
\text{for some $\tm_{b}''\neq m_{b}''$}
  \big\} \big| \leq
2^{-n\eps} \,.
\end{align}
\end{enumerate}
\end{lemma}

Then, define
\begin{align}
&\xvec_1(m_{b-1}')=\sqrt{n\gamma(\plimit-\delta)}\cdot \avec(m_{b-1}') \,, \nonumber\\
&\xvec'(m_{b}',m_{b}''|m_{b-1}')=  \rho\sqrt{\alpha \gamma^{-1}}\cdot \xvec_1(m_{b-1}')
+\beta\cdot \vvec(m_b',m_b'')  \,,
\end{align}
where
\begin{align}
 \beta\triangleq& \sqrt{n(1-\rho^2)\alpha(\plimit-\delta)} \,,\;
 \gamma\triangleq \nicefrac{\plimit_1}{\plimit} \,.
\label{eq:RYGbeta}
\end{align}
Note that $\norm{\xvec_1(m_{b-1}')}^2= n\gamma(\plimit-\delta)<n\plimit_1$, for all $m_{b-1}'\in [1:2^{nR'}]$. On the other hand, $\norm{\xvec'(m_{b}',m_{b}''|m_{b-1}')}^2$ could be greater than $n\alpha\plimit$ due to the possible correlation between $\xvec_1(m_{b-1}')$ and $\vvec(m_b',m_b'')$.

\emph{Encoding}:
Let $(m_1',m_1'',\ldots, m_{B-1}',m_{B-1}'')$ be a sequence of messages to be sent.
In block $b\in [1:B]$, if $\norm{\xvec'(m_{b}',m_{b}''|m_{b-1}')}^2$ $\leq n\alpha\plimit$, transmit
$(\xvec'(m_{b}',m_{b}''|m_{b-1}'),\xvec''(m_b'))$. Otherwise,
transmit $(\mathbf{0},\xvec''(m_b'))$.

\emph{Relay Encoding}:
In block $1$, the relay transmits $\xvec_1(1)$. 
 At the end of block $b\in [1:B-1]$, the relay receives $\yvec_{1,b}$, and
 finds an estimate $\bar{m}_b'=g_1(\yvec_{1,b})$. 
 In block $b+1$, the relay transmits $\xvec_1(\bar{m}_b')$.

\emph{Backward Decoding}: Once all blocks $(\yvec_b)_{b=1}^B$ are received, decoding is performed backwards.
Set $\hm_0'=\hm_0''\equiv 1$. 
For $b=B-1,B-2,\ldots,1$, find a unique $\hm_b'\in[1:2^{nR'}]$ such that
\begin{align}
\norm{\yvec_{b+1}- (1+\rho\sqrt{\alpha\gamma^{-1}})\xvec_1(\hm_{b}')} \leq \norm{ \yvec_{b+1}-(1+\rho\sqrt{\alpha\gamma^{-1}})\xvec_1(m_{b}')} \,,\; 
\text{for all $m_b'\in [1:2^{nR}]$} \,.
\end{align}
If there is more than one 
 such $\hm_b'\in[1:2^{nR'}]$,
 declare an error.

Then, the decoder uses $\hm_1',\ldots,\hm_{B-1}'$ as follows.
For $b=B-1,B-2,\ldots,1$, find a unique $\hm_b''\in[1:2^{nR''}]$ such that
\begin{align}
\norm{\yvec_{b}- \xvec_1(\hm_{b-1}')-\xvec'(\hm_{b}',\hm_{b}''|\hm_{b-1}')} \leq 
\norm{\yvec_{b}- \xvec_1(\hm_{b-1}')-\xvec'(\hm_{b}',m_{b}''|\hm_{b-1}')} \,,\; 
\text{for all $m_b''\in [1:2^{nR''}]$} \,.
\end{align}  
If there is more than one 
 such $\hm_b''\in[1:2^{nR''}]$, declare an error.

\emph{Analysis of Probability of Error}:
 Fix $\svec\in\Sset^n$, and let
\begin{align}
\cvec_0\triangleq \frac{\svec}{\norm{\svec}} \,.
\end{align}
The error event is bounded by the union of the following events. For $b\in [1:B-1]$, define
\begin{align}
&\Eset_{1}(b)=\{ \bar{M}_b'\neq M_b' \} \,,\; 
\Eset_{2}(b)=\{ \hM_b'\neq M_b' \} \,,\;
\Eset_{3}(b)=\{ \hM_b''\neq M_b'' \} \,.
\end{align}
 Then,  the conditional probability of error given the state sequence $\svec$ is bounded by
\begin{align}
P_{e|\svec}(\code) 
\leq& \sum_{b=1}^{B-1} \prob{\Eset_1(b)}+\sum_{b=1}^{B-1} \prob{\Eset_{2}(b) \cap \Eset_1^c(b)}
+ \sum_{b=1}^{B-1} \prob{\Eset_3(b) \cap \Eset_1^c(b-1)\cap \Eset_2^c(b) \cap \Eset_2^c(b-1) }  \,,
\label{eq:RYCompoundCerrBoundGauss}
\end{align}
with $\Eset_1(0)=\Eset_2(0)=\emptyset$, where the conditioning on $\Svec=\svec$ is omitted for convenience of notation.
Recall that we have defined $\code''$ 
 as a $(2^{nR'},n,\eps)$ code for the point to point gaussian channel $Y_1=X''+Z$. Hence, the first sum in the RHS of (\ref{eq:RYCompoundCerrBoundGauss}) is bounded by $B\cdot\eps$, which is arbitrarily small.

As for the erroneous decoding of $M_b'$ at the receiver, consider the following events,
\begin{align}
&\Eset_{2}(b)=\{
 \norm{ \Yvec_{b+1}- (1+\rho\sqrt{\alpha\gamma^{-1}})\xvec_1(\tm_{b}')} \leq \norm{ \Yvec_{b+1}-(1+\rho\sqrt{\alpha\gamma^{-1}})\xvec_1(M_{b}')} \,,\; 
\text{for some $\tm_b'\neq M_b'$}
 \} \,, \nonumber\\
&\Eset_{2,1}(b)=\{ |\langle \avec(M_{b}'),\cvec_0 \rangle| \geq \eta \} \,, \nonumber\\
&\Eset_{2,2}(b)=\{ |\langle \avec(M_{b}'),\vvec(M_{b+1}) \rangle| \geq \eta \} \nonumber\\
&\Eset_{2,3}(b)=\{ |\langle \vvec(M_{b+1}),\cvec_0 \rangle| \geq \eta \} \nonumber\\
&\widetilde{\Eset}_{2}(b)=\Eset_{2}(b)\cap \Eset_1^c(b)\cap \Eset_{2,1}^c(b)\cap \Eset_{2,2}^c(b)\cap \Eset_{2,3}^c(b) \,,
\label{eq:RYgE24T}
\end{align}
where $M_{b+1}=(M_{b+1}',M_{b+1}'')$. Then,
\begin{align}
\Eset_2(b)\cap\Eset_1^c(b)\subseteq& \Eset_{2,1}(b)\cup \Eset_{2,2}(b)\cup\Eset_{2,3}(b)\cup
(\Eset_{2}(b)\cap\Eset_1^c(b)) \nonumber\\
=& \Eset_{2,1}(b)\cup \Eset_{2,2}(b)\cup\Eset_{2,3}(b)\cup\widetilde{\Eset}_{2}(b) \,.
\end{align}
Hence, by the union of events bound, we have that
\begin{align}
\prob{\Eset_{2}(b) \cap \Eset_1^c(b)}\leq& \prob{\Eset_{2,1}(b)}+\prob{\Eset_{2,2}(b)}+\prob{\Eset_{2,3}(b)}
+\prob{\widetilde{\Eset}_{2}(b)} \,.
\label{eq:RYgE2nE1c}
\end{align}
By Lemma~\ref{lemm:RYgKey}, given $R'>-\frac{1}{2}\log(1-\eta^2)$, the first term is bounded by
\begin{align}
\prob{\Eset_{2,1}(b)}=& \prob{\langle \avec(M_{b}'),\cvec_0 \rangle \geq \eta}+
\prob{\langle \avec(M_{b}'),-\cvec_0 \rangle \geq \eta}
\nonumber\\
\leq& 2\cdot\frac{1}{2^{nR'}} \cdot 2^{n\left(R'+\frac{1}{2}\log(1-\eta^2)+\eps\right)}
\leq 2\cdot
2^{n\left(
-\frac{1}{2}\eta^2+\eps\right)} \,,
\label{eq:RYgE21}
\end{align}
 since $\log(1+t)\leq t$ for $t\in\mathbb{R}$. 
As $\eta^2\geq 8\eps$, the last expression tends to zero as $n\rightarrow\infty$. 
Similarly, $\prob{\Eset_{2,2}(b)}$ and $\prob{\Eset_{2,3}(b)}$ tend to zero as well. 
Moving to the fourth term in the RHS of (\ref{eq:RYgE2nE1c}), observe that for a sufficiently small $\eps$ and $\eta$, the event $\Eset_{2,2}^c(b)$ implies that
 $\norm{\xvec'(M_{b+1}|M_{b}')}^2\leq n\alpha\plimit$, while the event $\Eset_{1}^c(b)$ means that 
$\bar{M}_b'=M_b'$. Hence, the encoder transmits
$(\xvec'(M_{b+1}|M_{b}'),\xvec''(M_{b+1}'))$, the relay transmits $\xvec_1(M_{b}')$, and we have that
\begin{align}
&\norm{ \Yvec_{b+1}- (1+\rho\sqrt{\alpha\gamma^{-1}})\xvec_1(\tm_{b}')}^2 - \norm{ \Yvec_{b+1}-(1+\rho\sqrt{\alpha\gamma^{-1}})\xvec_1(M_{b}')}^2 \nonumber\\
=&\norm{(1+\rho\sqrt{\alpha\gamma^{-1}})\xvec_1(M_{b}')+\beta\vvec(M_{b+1})+\svec -(1+\rho\sqrt{\alpha\gamma^{-1}})\xvec_1(\tm_{b}')}^2-\norm{\beta\vvec(M_{b+1})+\svec}^2
\nonumber\\
=& 2(1+\rho\sqrt{\alpha\gamma^{-1}})^2\left( \frac{1}{2}\norm{\xvec_1(M_{b}')}^2  +\frac{1}{2}\norm{\xvec_1(\tm_{b}')}^2
 -\langle \xvec_1(\tm_{b}'), \xvec_1(M_{b}') \rangle
\right) 
\nonumber\\
&+2(1+\rho\sqrt{\alpha\gamma^{-1}})\left(  \langle \xvec_1(M_{b}'),\beta\vvec(M_{b+1})+\svec \rangle
-  \langle \xvec_1(\tm_{b}'), \beta\vvec(M_{b+1})+\svec \rangle \right)
\end{align}
Then, since $\norm{\xvec_1(m_b')}^2=n\gamma(\plimit-\delta)$ for all $m_b'\in [1:2^{nR'}]$, we have that
\begin{align}
\Eset_{2}(b)\cap \Eset_{1}^c(b)\cap\Eset_{2,2}^c(b)  
\subseteq& 
\{ (1+\rho\sqrt{\alpha\gamma^{-1}})\langle \xvec_1(\tm_{b}'), \xvec_1(M_{b}') \rangle
+\langle \xvec_1(\tm_{b}'), \beta\vvec(M_{b+1})+\svec  \rangle \geq \nonumber\\&\quad
  n(1+\rho\sqrt{\alpha\gamma^{-1}})\gamma(\plimit-\delta)+
\langle \xvec_1(M_{b}'),\beta\vvec(M_{b+1})+\svec \rangle
 \,,\;\text{for some $\tm_b'\neq M_b'$}\} \,.
\label{eq:RYgE1equiv1}
\end{align}
Observe that for sufficiently small $\eps$ and $\eta$, the event 
$\Eset_{2,1}^c(b)\cap \Eset_{2,2}^c(b)\cap \Eset_{2,3}^c(b)$ implies that 
\begin{align}
&\langle \xvec_1(M_{b}'),\beta\vvec(M_{b+1})+\svec \rangle \geq -\delta \,,
\label{eq:RYgE0i1}
\intertext{and}
&\norm{\beta\vvec(M_{b+1})+\svec}^2 \leq n[(1-\rho^2)\alpha\plimit+\Lambda] \,.
\label{eq:RYgE0i2}
\end{align}
Hence, by (\ref{eq:RYgE1equiv1}) and (\ref{eq:RYgE0i1}),
\begin{align}
&\widetilde{\Eset}_{2}(b)=\Eset_{2}(b)\cap \Eset_1^c(b)\cap \Eset_{2,1}^c(b)\cap \Eset_{2,2}^c(b)\cap \Eset_{2,3}^c(b)\nonumber\\
\subseteq&\{ 
(1+\rho\sqrt{\alpha\gamma^{-1}})\langle \xvec_1(\tm_{b}'), \xvec_1(M_{b}') \rangle
+\langle \xvec_1(\tm_{b}'), \beta\vvec(M_{b+1})+\svec  \rangle \geq  
n(1+\rho\sqrt{\alpha\gamma^{-1}})\gamma(\plimit-2\delta) \,,\;
\text{for some $\tm_b'\neq M_b'$} \} \,.
\label{eq:RYgE1b1}
\end{align}
Dividing both sides of the inequality by $n(1+\rho\sqrt{\alpha\gamma^{-1}})$, we obtain
\begin{align}
\widetilde{\Eset}_{2}(b)
\subseteq \left\{  
\frac{1}{n}\langle \xvec_1(\tm_{b}'), \xvec_1(M_{b}') \rangle
+\frac{\langle \xvec_1(\tm_{b}'), \beta\vvec(M_{b+1})+\svec  \rangle}{n(1+\rho\sqrt{\alpha\gamma^{-1}})} 
\geq   \gamma(\plimit-2\delta) \,,\; 
\text{for some $\tm_b'\neq M_b'$} \right\} \,.
\label{eq:RYgE1b3}
\end{align}
Next, we partition the set of values of $\frac{1}{n}\langle \xvec_1(\tm_{b}'), \xvec_1(M_{b}') \rangle
$ to $K$ bins. 
Let $\tau_1<\tau_2<\cdots<\tau_K$ be such partition,
 where
\begin{align}
&\tau_1=\gamma(\plimit-2\delta)-\frac{\sqrt{(\plimit-\delta)[(1-\rho^2)\alpha\plimit+\Lambda]}}{1+\rho\sqrt{\alpha\gamma^{-1}}}  \,,\;
\tau_K= \gamma(\plimit-3\delta) \,,\nonumber \\
&\tau_{k+1}-\tau_k \leq \gamma\cdot\delta  \,,\;\text{for $k=[1:K-1]$}\,,
\end{align}
where $K$ is a finite constant which is independent of $n$, as in Lemma~\ref{lemm:RYgKey}.
By (\ref{eq:RYgE0i2}) and (\ref{eq:RYgE1b3}), given the event $\widetilde{\Eset}_{2}(b)
$, we have that 
\begin{align}
\frac{1}{n}\langle \xvec_1(\tm_{b}'), \xvec_1(M_{b}') \rangle \geq \tau_1 >0\,,
\end{align}
where the last inequality is due to (\ref{eq:RYGassump2}),  for sufficiently small $\delta>0$.
To see this, observe that the inequality in (\ref{eq:RYGassump2}) is  strict, and it implies that
\begin{align}
\sqrt{\gamma} \cdot (\sqrt{\gamma\plimit}+\rho\sqrt{\alpha\plimit})>\sqrt{(1-\rho^2)\alpha\plimit+\Lambda} \,.
\end{align}
Hence, for sufficiently small $\delta>0$, $\tau_1>0$ as
\begin{align}
\tau_1=\frac{\sqrt{\plimit-2\delta}}{1+\rho\sqrt{\alpha\gamma^{-1}}}\cdot
\left(
\sqrt{\gamma}(\sqrt{\gamma(\plimit-2\delta)}+\rho\sqrt{\alpha
(\plimit-2\delta)})
-\sqrt{\frac{\plimit-\delta}{\plimit-2\delta}[(1-\rho^2)\alpha\plimit+\Lambda]}
\right)
\,.
\end{align}
Furthermore,  if $\tau_k\leq  \frac{1}{n}\langle \xvec_1(\tm_{b}'), \xvec_1(M_{b}') \rangle < \tau_{k+1}$, then
\begin{align}
&\frac{\langle \xvec_1(\tm_{b}'), \beta\vvec(M_{b+1})+\svec  \rangle}{n(1+\rho\sqrt{\alpha\gamma^{-1}})}  \geq 
\gamma(\plimit-2\delta)-\tau_{k+1}
\geq \gamma(\plimit-3\delta)-\tau_{k} \,. 
\end{align}
Thus, 
\begin{align}
\prob{\widetilde{\Eset}_{2}(b)} 
\leq& 
 \sum_{k=1}^{K-1} \Pr\bigg(\frac{1}{n}|\langle \xvec_1(\tm_{b}'), \xvec_1(M_{b}') \rangle| \geq \tau_{k} \,,\;  
\frac{|\langle \xvec_1(\tm_{b}'), \beta\vvec(M_{b+1})+\svec  \rangle|}{n(1+\rho\sqrt{\alpha\gamma^{-1}})}\geq
\gamma(\plimit-3\delta)-\tau_{k}
 \,,\; \text{for some $\tm_b'\neq M_b'$} \bigg) 
\nonumber\\&
+ \prob{ \frac{1}{n}| \langle \xvec_1(\tm_{b}'), \xvec_1(M_{b}') \rangle |\geq \tau_{K} 
\,,\;\text{for some $\tm_b'\neq M_b'$} } \,.
\label{eq:RYgE1b5}
\end{align}
By (\ref{eq:RYgE0i2}), this can be further bounded by
\begin{align}
\prob{\widetilde{\Eset}_{2}(b)} 
\leq&  \sum_{k=1}^{K} \Pr\bigg(|\langle \avec(\tm_{b}'), \avec(M_{b}') \rangle| \geq \theta_k \,,\; 
|\langle \avec(\tm_{b}'), \cvec'(M_{b+1}) \rangle| \geq \mu_k
 \,,\;  \text{for some $\tm_b'\neq M_b'$} \bigg)  \,,
\label{eq:RYgE24b4}
\end{align}
where
\begin{align}
\cvec'(m_{b+1})\triangleq\frac{\beta\vvec(m_{b+1})+\svec}{\norm{\beta\vvec(m_{b+1})+\svec}} \,,
\end{align}
and 
\begin{align}
&\theta_k \triangleq \frac{\tau_{k}}{\gamma(\plimit-\delta)} \,,\;
\zeta_k\triangleq \frac{(1+\rho\sqrt{\alpha\gamma^{-1}})\left(\gamma(\plimit-3\delta)-\tau_{k}\right)}{\sqrt{\gamma(\plimit-\delta)((1-\rho^2)\alpha\plimit+\Lambda)}} \,,\;\text{for $k\in [1:K-1]$} \,;\;
 \theta_K\triangleq \frac{\tau_{K}}{\plimit-\delta} \,,\; \zeta_K=0 \,.
\end{align}

By Lemma~\ref{lemm:RYgKey}, the RHS of (\ref{eq:RYgE24b4}) tends to zero as $n\rightarrow\infty$ provided that
\begin{align}
\theta_k \geq \eta \;\text{and }\; \theta_k^2+\zeta_k^2>1+\eta-e^{-2R'} \,,\;\text{for $k=[1:K]$}\,.
\end{align}
For sufficiently small $\eps$ and $\eta$, we have that $\eta\leq \theta_1=\frac{\tau_1}{\gamma(\plimit-\delta)}$, hence the first condition is met. Then,  observe that the second condition is equivalent to
$G(\tau_k)>1+\eta-e^{-2R'}$, for $k\in [1:K-1]$, where 
\begin{align}
G(\tau)=(A\tau)^2+D^2(L-\tau)^2 \,,\;
\text{with }
A=\frac{1}{\gamma(\plimit-\delta)} \,,\; D=\frac{1+\rho\sqrt{\alpha\gamma^{-1}}}{\sqrt{\gamma(\plimit-\delta)((1-\rho^2)\alpha\plimit+\Lambda)  }} \,,\; L=\gamma(\plimit-3\delta) \,.
\end{align}
By differentiation, we have that the minimum value of this function is given by
$
\min_{\tau_1\leq \tau\leq\tau_K} G(\tau)=\frac{A^2 D^2 L^2}{A^2+D^2}= \frac{D^2}{A^2+D^2}-\delta_1 
$, 
where $\delta_1\rightarrow 0$ as $\delta\rightarrow 0$. 
Thus, the RHS of (\ref{eq:RYgE24b4}) tends to zero as $n\rightarrow\infty$, provided that
\begin{align}
R'<&-\frac{1}{2}\log\left(1+\eta-\frac{D^2}{A^2+D^2}+\delta_1 \right)
=
-\frac{1}{2}\log\left(\eta+\delta_1+\frac{(1-\rho^2)\alpha\plimit+\Lambda}{(\gamma+\alpha+2\rho\sqrt{\alpha\gamma})\plimit+\Lambda-\delta\gamma(1+\rho\sqrt{\alpha\gamma^{-1}})^2} \right) \,.
\end{align}
This is satisfied for
$
R'=\inR'_\alpha(\avrc)-\delta'$,  with
\begin{align}
\inR'_\alpha(\avrc)=\frac{1}{2}\log\left(\frac{(\gamma+\alpha+2\rho\sqrt{\alpha\gamma})\plimit+\Lambda}{(1-\rho^2)\alpha\plimit+\Lambda} \right)=-\frac{1}{2}\log\left(\frac{(1-\rho^2)\alpha\plimit+\Lambda}{(\gamma+\alpha+2\rho\sqrt{\alpha\gamma})\plimit+\Lambda} \right) \,.
\end{align}
and arbitrary  $\delta'>0$,
if $\eta$ and $\delta$ are sufficiently small.

Moving to the error event for $M_b''$, consider the events
\begin{align}
&\Eset_{3}(b)=\Big\{
  \norm{ \Yvec_b- \xvec_1(\hM_{b-1}')-\xvec'(\hM_b',\tm_b''|\hM_{b-1}')} \leq 
\norm{ \Yvec_b- \xvec_1(\hM_{b-1}')-\xvec'(\hM_b',M_b''|\hM_{b-1}')} \,,\;
\text{for some $\tm_b''\neq M_b''$}
 \Big\} \,, \nonumber\\
&\Eset_{3,1}(b)=\{ |\langle \vvec(M_{b}),\cvec_0) \rangle| \geq \eta \} \nonumber\\
&\Eset_{3,2}(b)=\{ |\langle \avec(M_{b-1}'), \vvec(M_{b})) \rangle| \geq \eta \} \nonumber\\
&\widetilde{\Eset}_{3}(b)=\Eset_{3}(b) \cap \Eset_1^c(b)\cap \Eset_2^c(b) \cap \Eset_2^c(b-1) \cap
 \Eset_{3,1}^c(b)\cap \Eset_{3,2}^c(b) \,,
\label{eq:RYgE33T}
\end{align}
where $M_{b}=(M_{b}',M_{b}'')$. Then,
\begin{align}
\Eset_3(b) \cap \Eset_1^c(b-1)\cap \Eset_2^c(b) \cap \Eset_2^c(b-1) 
\subseteq& 
 \Eset_{3,1}(b)\cup \Eset_{3,2}(b)\cup (\Eset_{3}(b)\cap\Eset_1^c(b)\cap\Eset_1^c(b)\cap \Eset_2^c(b) \cap \Eset_2^c(b-1)) \nonumber\\
=& \Eset_{3,1}(b)\cup\Eset_{3,2}(b)\cup \widetilde{\Eset}_{3}(b) \,.
\end{align}
Hence, by the union of events bound, we have that
\begin{align}
\prob{\Eset_3(b) \cap \Eset_1^c(b-1)\cap \Eset_2^c(b) \cap \Eset_2^c(b-1) }\leq
\prob{\Eset_{3,1}(b)}+\prob{\Eset_{3,2}(b)}+\prob{\widetilde{\Eset}_{3}(b)} \,.
\label{eq:RYgE3nE2c}
\end{align}
By Lemma~\ref{lemm:RYgKey}, given $R''>-\frac{1}{2}\log(1-\eta^2)$, the first term is bounded by
\begin{align}
\prob{\Eset_{3,1}(b)}=& \prob{\langle \vvec(M_{b}',M_b''),\cvec_0 \rangle \geq \eta}+
\prob{\langle \vvec(M_{b}',M_b''),-\cvec_0 \rangle \geq \eta}
\nonumber\\
\leq& 2\cdot\frac{1}{2^{nR''}} \cdot 2^{n\left(R''+\frac{1}{2}\log(1-\eta^2)+\eps\right)}
\leq 2\cdot
2^{n\left(
-\frac{1}{2}\eta^2+\eps\right)} \,,
\label{eq:RYgE31}
\end{align}
 since $\log(1+t)\leq t$ for $t\in\mathbb{R}$. 
As $\eta^2\geq 8\eps$, the last expression tends to zero as $n\rightarrow\infty$. 
Similarly, $\prob{\Eset_{3,2}(b)}$ tends to zero as well.
As for the third term in the RHS of (\ref{eq:RYgE3nE2c}), observe that for a sufficiently small $\eps$ and $\eta$, the event $\Eset_{3,2}^c(b)$ implies that
 $\norm{\xvec'(M_{b}|M_{b-1}')}^2\leq n\alpha\plimit$, and as the event $\Eset_1^c(b)\cap\Eset_2^c(b)\cap\Eset_2^c(b-1)$ occurs, we have that $\tM_{b-1}'=M_{b-1}'$, $\hM_b'=M_b'$, and $\hM_{b-1}'=M_{b-1}'$. Hence, the encoder transmits
$(\xvec'(M_{b}|M_{b-1}'),\xvec''(M_{b+1}'))$, the relay transmits $\xvec_1(M_{b}')$, and we have that 
\begin{align}
& \norm{\Yvec_b- \xvec_1(\hM_{b-1}')-\xvec'(\hM_b',\tm_b''|\hM_{b-1}')}^2 -
\norm{ \Yvec_b- \xvec_1(\hM_{b-1}')-\xvec'(\hM_b',M_b''|\hM_{b-1}')}^2 \nonumber\\
=&\norm{\xvec'(M_{b}',M_b''|M_{b-1}')+\svec -\xvec'(M_b',\tm_b''|M_{b-1}')}^2-\norm{\svec}^2
\nonumber\\
=&\norm{\beta\vvec(M_{b}',M_b'')-\beta\vvec(M_b',\tm_b'')+\svec }^2-\norm{\svec}^2
\nonumber\\
=& 2\beta^2
+  2\beta\langle \vvec(M_{b}',M_b''),\svec \rangle
-  2\beta^2\langle  \vvec(M_b',\tm_b''),\vvec(M_{b}',M_b'') \rangle
-  2\beta\langle \vvec(M_b',\tm_b''), \svec \rangle
\end{align}
Then,  we have that
\begin{align}
&\Eset_{3}(b)\cap \Eset_1^c(b)\cap\Eset_2^c(b)\cap\Eset_2^c(b-1)\cap   \Eset_{3,2}^c(b)
\nonumber\\
\subseteq& \{ \beta\langle  \vvec(M_b',\tm_b''),\vvec(M_{b}',M_b'') \rangle
+\langle \vvec(M_b',\tm_b''), \svec \rangle \geq 
 \beta+\langle \vvec(M_{b}',M_b''),\svec \rangle \,,\; 
\text{for some $\tm_b''\neq M_b''$}\} \,.
\label{eq:RYgE3equiv1}
\end{align}
Observe that for sufficiently small $\eps$ and $\eta$, the event 
$\Eset_{3,1}^c(b)$ implies that 
$
\langle \vvec(M_{b}',M_b''),\svec \rangle \geq -\beta\cdot\delta
$. 
Hence, by (\ref{eq:RYgE3equiv1}), 
\begin{align}
\widetilde{\Eset}_{3}
\subseteq&\{ \langle  \vvec(M_b',\tm_b''),\vvec(M_{b}',M_b'') \rangle
+\frac{1}{\beta}\langle \vvec(M_b',\tm_b''), \svec \rangle \geq 
 1-\delta \,,\; \text{for some $\tm_b''\neq M_b''$}\}
\label{eq:RYgE2b1}
\end{align}
Next, we partition the set of values of $\langle  \vvec(M_b',\tm_b''),\vvec(M_{b}',M_b'') \rangle$. Let 
$\tau_1''<\tau_2''<\cdots<\tau_K''$ be such partition, where
\begin{align}
&\tau_1''=1-\delta-\frac{\sqrt{n\Lambda}}{\beta} \,,\;
\tau_K''= 1-2\delta \,,\\
&\tau_{k+1}''-\tau_k'' \leq \delta \,,\;\text{for $k=[1:K-1]$}\,.
\end{align}
By (\ref{eq:RYgE2b1}), given the event $\widetilde{\Eset}_{3}(b)$, we have that 
$
\langle  \vvec(M_b',\tm_b''),$ $\vvec(M_{b}',M_b'')  \rangle$ $ \geq$ $ \tau_1'' 
$, where $\tau_1''>0$ due to (\ref{eq:RYGassump1}) and (\ref{eq:RYGbeta}). 
Now, if $
  \langle  \vvec(M_b',\tm_b''),\vvec(M_{b}',M_b'') \rangle 
	$ is in the interval $[\tau_k'',\tau_{k+1}'']$, then it follows that
$
\langle \vvec(M_b',\tm_b''), \svec \rangle  
\geq \beta(1-2\delta-\tau_{k}'') 
$. 
Hence, 
\begin{align}
\prob{\widetilde{\Eset}_{3}(b)} 
\leq&  \sum_{k=1}^{K-1} \Pr\bigg( \langle  \vvec(M_b',\tm_b''),\vvec(M_{b}',M_b'') \rangle \geq \tau_{k}'' \,,\;
\langle \vvec(M_b',\tm_b''), \cvec_0 \rangle\geq
\frac{\beta(1-2\delta-\tau_{k}'')}{\sqrt{n\Lambda}}
 \,,\;  \text{for some $\tm_b''\neq M_b''$} \bigg) 
\nonumber\\&
+ \Pr\big(  \langle  \vvec(M_b',\tm_b''),\vvec(M_{b}',M_b'') \rangle \geq \tau_{K}'' 
\,,\;\text{for some $\tm_b''\neq M_b''$} \big) \,.
\label{eq:RYgE2b4}
\end{align}
By part 2 of Lemma~\ref{lemm:RYgKey}, the RHS of (\ref{eq:RYgE2b4}) tends to zero as $n\rightarrow\infty$ provided that
\begin{align}
\tau_k'' \geq \eta \;\text{and }\; \tau_k''^2+\zeta_k''^2>1+\eta-e^{-2R''} \,,\;\text{for $k=[1:K]$}\,,
\end{align}
where $\zeta_k''\triangleq \frac{\beta(1-2\delta-\tau_{k}'')}{\sqrt{n\Lambda}}$ for $k\in [1:K-1]$ and
$\zeta_K''=0$. 
For sufficiently small $\eps$ and $\eta$, we have that $\eta\leq \tau_1''$, hence the first condition is met. 
By differentiation, we have that the minimum value of the function $\tilde{G}(\tau)=
\tau^2+\frac{\beta^2}{n\Lambda}(1-2\delta-\tau)^2$
 is given by
$
\frac{\beta^2(1-2\delta)^2}{\beta^2+n\Lambda} 
$. 
Thus, the RHS of (\ref{eq:RYgE2b4}) tends to zero as $n\rightarrow\infty$, provided that
\begin{align}
R''<&-\frac{1}{2}\log\left(1+\eta-\frac{\beta^2(1-2\delta)^2}{\beta^2+n\Lambda} \right)
<
-\frac{1}{2}\log\left(\eta+\frac{\Lambda}{(1-\rho^2)\alpha(\plimit-\delta)+\Lambda} \right) \,.
\end{align}
This is satisfied for
$
R''=\inR''_\alpha(\avrc)-\delta''$, with
\begin{align}
\inR''_\alpha(\avrc)=\frac{1}{2}\log\left(\frac{(1-\rho^2)\alpha\plimit+\Lambda}{\Lambda} \right)=-\frac{1}{2}\log\left(\frac{\Lambda}{(1-\rho^2)\alpha\plimit+\Lambda} \right)
\end{align}
for an arbitrary $\delta''>0$, 
if $\eta$ and $\delta$ are sufficiently small.

We have thus shown achievability of every rate
\begin{align}
R<\min\left\{  \inR'_\alpha(\avrc) +\inR''_\alpha(\avrc) , 
\frac{1}{2}\log\left( 1+\frac{(1-\alpha)\plimit}{\sigma^2} \right)+\inR''_\alpha(\avrc) \right\} \,,
\end{align}
where
\begin{align}
\inR'_\alpha(\avrc) +\inR''_\alpha(\avrc)=&\frac{1}{2}\log\left(\frac{(\gamma+\alpha+2\rho\sqrt{\alpha\gamma})\plimit+\Lambda}{(1-\rho^2)\alpha\plimit+\Lambda} \right)+\frac{1}{2}\log\left(\frac{(1-\rho^2)\alpha\plimit+\Lambda}{\Lambda} \right)
\nonumber\\
=&\frac{1}{2} \log\left( 1+  \frac{\plimit_1+\alpha\plimit+2\rho\sqrt{\alpha\plimit\cdot \plimit_1}}{\Lambda} \right) 
\end{align}
(see (\ref{eq:RYGbeta})). 
This completes the proof of the lower bound. 
%
\qed

\subsection{Upper Bound}
Let $R>0$ be an achievable rate. Then, there exists a sequence of $(2^{nR},n,\eps_n^*)$ codes 
$\code_n=(\fvec,\fvec_1,g)$ for the Gaussian AVRC $\avrc$ with SFD such that $\eps_n^*\rightarrow 0$ as $n\rightarrow\infty$, where the encoder consists of a pair $\fvec=(\fvec',\fvec'')$, with $\fvec':[1:2^{nR}]\rightarrow\mathbb{R}^n$ and $\fvec'':[1:2^{nR}]\rightarrow\mathbb{R}^n$.
Assume without loss of generality that the codewords have zero mean, \ie
\begin{align}
&\frac{1}{2^{nR}}\sum_{m=1}^{2^{nR}}\frac{1}{n} \sum_{i=1}^n \mathrm{f}_i(m)=0 \,,\;
\nonumber\\
&\int_{-\infty}^\infty \,d\yvec_1\cdot \frac{1}{2^{nR}}\sum_{m\in [1:2^{nR}]} 
P_{\Yvec_1|M}(\yvec_1|m)\cdot\frac{1}{n} \sum_{i=1}^n \mathrm{f}_{1,i}(y_{1,1},y_{1,2},\ldots,y_{1,i-1})=0 \,,
\end{align}
where
$
P_{\Yvec_1|M}(\yvec_1|m)=\frac{1}{(2\pi \sigma^2)^{\nicefrac{n}{2}}}e^{-\norm{\yvec_1-\fvec''(m)}^2/2\sigma^2} 
$. 
If this is not the case, redefine the code such that the mean is subtracted from each codeword.
 Then, define
\begin{align}
&\alpha\triangleq \frac{1}{n\plimit}\cdot\frac{1}{2^{nR}}\sum_{m\in [1:2^{nR}]} 
\norm{\fvec'(m)}^2 \,,
\nonumber\\
&\alpha_1\triangleq
 \frac{1}{n\plimit_1} \cdot\frac{1}{2^{nR}}\sum_{m\in [1:2^{nR}]}
\int_{-\infty}^\infty \,d\yvec_1\cdot 
P_{\Yvec_1|M}(\yvec_1|m)\cdot
\norm{\fvec_{1}(\yvec_1)}^2  \,.
\nonumber\\
&\rho\triangleq
 \frac{1}{n\sqrt{\alpha\plimit\cdot\alpha_1\plimit_1}}
\int_{-\infty}^\infty \,d\yvec_1\cdot \frac{1}{2^{nR}}\sum_{m\in [1:2^{nR}]} 
P_{\Yvec_1|M}(\yvec_1|m)\cdot
\langle \fvec'(m), \fvec_{1}(\yvec_1) \rangle \,,
\label{eq:RYgRhoAdef}
\end{align}
Since the code satisfies the input constraints $\plimit$ and $\plimit_1$, we have that $\alpha$, $\alpha_1$ and $\rho$ are in the interval $[0,1]$. 

First, we show that if 
\begin{align}
\Lambda>\plimit_1+\alpha\plimit+2\rho\sqrt{\alpha\plimit\cdot\plimit_1}+\delta \,,
\label{eq:RYgConvAssumpL}
\end{align}
 then the capacity is zero, where $\delta>0$ is arbitrarily small.
Consider the following jamming strategy. The jammer draws a message $\tM\in [1:2^{nR}]$ uniformly at random, and then,  generates a sequence $\widetilde{\Yvec}_1\in\mathbb{R}^n$ distributed according to $P_{\Yvec_1|M}(\tilde{\yvec}_1|\tm)$. Let $\widetilde{\Svec}=\fvec'(\tM)+\fvec_1(\widetilde{\Yvec}_1)$. If $\frac{1}{n}\norm{\widetilde{\Svec}}^2\leq \Lambda$,  the jammer chooses $\widetilde{\Svec}$ to be the state sequence. Otherwise, let the state sequence consist of all zeros.
Observe that
\begin{align}
\E\norm{\widetilde{\Svec}}^2=&\E\norm{\fvec'(\tM)+\fvec_1(\widetilde{\Yvec}_1)}^2
\nonumber\\
=&\E\norm{\fvec'(\tM)}^2+\E\norm{\fvec_1(\widetilde{\Yvec}_1)}^2+2\E\langle \fvec(\tM), \fvec_{1}(\widetilde{\Yvec}_1) \rangle
\nonumber\\
=& n(\alpha\plimit+\alpha_1\plimit_1+2\rho\sqrt{\alpha\plimit\cdot \alpha_1\plimit_1})
\nonumber\\
\leq& n(\alpha\plimit+\plimit_1+2\rho\sqrt{\alpha\plimit\cdot \plimit_1}) < n(\Lambda-\delta) \,.
\label{eq:RYgtSnLd}
\end{align}
where the second equality is due to (\ref{eq:RYgRhoAdef}), 
and the last inequality is due to (\ref{eq:RYgConvAssumpL}).
Thus, by Chebyshev's inequality, there exists $\kappa>0$ such that
\begin{align}
\prob{\frac{1}{n}\norm{\widetilde{\Svec}}^2\leq \Lambda}\geq \kappa \,.
\label{eq:RYgConvEps1}
\end{align}
The state sequence $\Svec$ is then distributed according to
\begin{align}
&P_{\Svec|\{\frac{1}{n}\norm{\widetilde{\Svec}}^2\leq \Lambda\}}(\svec)= \frac{1}{2^{nR}}\sum_{\tm\in [1:2^{nR}]} \int_{\tilde{\yvec}_1: \fvec'(\tm)+\fvec_1(\tilde{\yvec}_1)=\svec} \,d\yvec_1
P_{\Yvec_1|M}(\yvec_1|m) \,, \nonumber \\
&\cprob{\Svec=\mathbf{0}}{\frac{1}{n}\norm{\widetilde{\Svec}}^2> \Lambda}=1 \,.
\label{eq:RYgQsnSymm}
\end{align}

Assume to the contrary that a positive rate can be achieved when the channel is governed by such state sequence, hence the size of the message set is at least $2$, \ie
$
\dM\triangleq 2^{nR} \geq 2 
$. 
The probability of error is then bounded by
\begin{align}
\err(q,\code)=&\int_{-\infty}^\infty \,d\svec\cdot q(\svec)P_{e|\svec}^{(n)}(\code)
\geq \prob{\frac{1}{n}\norm{\widetilde{\Svec}}^2\leq \Lambda}\cdot  \int_{\svec:\frac{1}{n}\norm{\svec}^2\leq\Lambda} \,d\svec \cdot P_{\Svec|\{\frac{1}{n}\norm{\widetilde{\Svec}}^2\leq \Lambda\}}(\svec)\cdot P_{e|\svec}^{(n)}(\code)
\nonumber\\
\geq& \kappa\cdot \int_{\svec:\frac{1}{n}\norm{\svec}^2\leq\Lambda} \,d\svec \cdot P_{\Svec|\{\frac{1}{n}\norm{\widetilde{\Svec}}^2\leq \Lambda\}}(\svec)\cdot P_{e|\svec}^{(n)}(\code)
\label{eq:RYGconvB1}
\end{align}
where the inequality holds by (\ref{eq:RYgConvEps1}).
Next, we have that
\begin{align}
P_{e|\svec}^{(n)}(\code) 
=& \frac{1}{\dM} \sum_{m=1}^{\dM}
\int_{-\infty   }^{\infty} \,d\yvec_1\cdot P_{\Yvec_1|M}(\yvec_1|m) \cdot
\mathds{1}\left\{\yvec_1: g(\fvec'(m)+\fvec_1(\yvec_1)+\svec)\neq m \right\}  \,,
\label{eq:RYgConvCerr}
\end{align}
where we define the indicator function $G(\yvec_1)=\mathds{1}\{\yvec_1\in\Aset \}$ such that
$G(\yvec_1)=1$ if $\yvec_1\in\Aset$, and $G(\yvec_1)=0$ otherwise.
 Substituting (\ref{eq:RYgQsnSymm}) and (\ref{eq:RYgConvCerr}) into (\ref{eq:RYGconvB1}) yields
\begin{align}
\err(q,\code) 
\geq& \kappa\cdot  \int_{\svec:\frac{1}{n}\norm{\svec}^2\leq\Lambda} \,d\svec\cdot \frac{1}{\dM}\sum_{\tm=1}^{\dM} \int_{\tilde{\yvec}_1: \fvec'(\tm)+\fvec_1(\tilde{\yvec}_1)=\svec} \,d\tilde{\yvec}_1\cdot P_{\Yvec_1|M}(\tilde{\yvec}_1|\tm)
\nonumber\\&\times
 \frac{1}{\dM} \sum_{m=1}^{\dM}
\int_{-\infty   }^{\infty} \,d\yvec_1\cdot P_{\Yvec_1|M}(\yvec_1|m) \cdot
\mathds{1}\left\{\yvec_1: g(\fvec'(m)+\fvec_1(\yvec_1)+\svec)\neq m \right\} 
 \,.
\end{align}
Eliminating $\svec=\fvec'(\tm)+\fvec_1(\tilde{\yvec}_1)$, and adding the constraint $\norm{\fvec'(m)+\fvec_1(\yvec_1)}^2\leq\Lambda$, 
we obtain the following,
\begin{align}
\err(q,\code) 
\geq&\frac{\kappa}{\dM^2}\sum_{m=1}^{\dM}\sum_{\tm=1}^{\dM} \int_{\substack{(\yvec_1\,,\;\tilde{\yvec}_1) \,:\; 
 \frac{1}{n}\norm{
\fvec'(m)+\fvec_1(\yvec_1)}^2\leq\Lambda\,, \\
\frac{1}{n}\norm{
\fvec'(\tm)+\fvec_1(\tilde{\yvec}_1)}^2\leq\Lambda  
}} \,d \yvec_1 \,d\tilde{\yvec}_1\cdot P_{\Yvec_1|M}(\yvec_1|m) P_{\Yvec_1|M}(\tilde{\yvec}_1|\tm)
\nonumber\\& 
\nonumber\\
&\times
\mathds{1}\left\{\yvec_1: g(\fvec'(m)+\fvec_1(\yvec_1)+\fvec'(\tm)+\fvec_1(\tilde{\yvec}_1))\neq m \right\} \,.
\end{align}
Now, by interchanging the summation variables $(m,\yvec_1)$ and $(\tm,\tilde{\yvec}_1)$, we have that
\begin{align}
\err(q,\code) 
\geq&\frac{\kappa}{2\dM^2}\sum_{m=1}^{\dM}\sum_{\tm=1}^{\dM} \int_{\substack{(\yvec_1\,,\;\tilde{\yvec}_1) \,:\; 
 \frac{1}{n}\norm{
\fvec'(m)+\fvec_1(\yvec_1)}^2\leq\Lambda\,, \\
\frac{1}{n}\norm{
\fvec'(\tm)+\fvec_1(\tilde{\yvec}_1)}^2\leq\Lambda  
}} \,d \yvec_1 \,d\tilde{\yvec}_1\cdot P_{\Yvec_1|M}(\yvec_1|m) P_{\Yvec_1|M}(\tilde{\yvec}_1|\tm)
\nonumber\\& 
\nonumber\\
&\times
\mathds{1}\left\{\yvec_1: g(\fvec'(m)+\fvec_1(\yvec_1)+\fvec'(\tm)+\fvec_1(\tilde{\yvec}_1))\neq m \right\}
\nonumber\\
+&\frac{\kappa}{2\dM^2}\sum_{m=1}^{\dM}\sum_{\tm=1}^{\dM} \int_{\substack{(\yvec_1\,,\;\tilde{\yvec}_1) \,:\; 
 \frac{1}{n}\norm{
\fvec'(m)+\fvec_1(\yvec_1)}^2\leq\Lambda\,, \\
\frac{1}{n}\norm{
\fvec'(\tm)+\fvec_1(\tilde{\yvec}_1)}^2\leq\Lambda  
}} \,d \yvec_1 \,d\tilde{\yvec}_1\cdot P_{\Yvec_1|M}(\yvec_1|m) P_{\Yvec_1|M}(\tilde{\yvec}_1|\tm)
\nonumber\\& 
\nonumber\\
&\times
\mathds{1}\left\{\yvec_1: g(\fvec'(m)+\fvec_1(\yvec_1)+\fvec'(\tm)+\fvec_1(\tilde{\yvec}_1))\neq \tm \right\}
 \,.
\end{align}
Thus,
\begin{align}
\err(q,\code) 
\geq&\frac{\kappa}{2\dM^2}\sum_{m=1}^{\dM}\sum_{\tm\neq m} \int_{\substack{(\yvec_1\,,\;\tilde{\yvec}_1) \,:\; 
 \frac{1}{n}\norm{
\fvec'(m)+\fvec_1(\yvec_1)}^2\leq\Lambda\,, \\
\frac{1}{n}\norm{
\fvec'(\tm)+\fvec_1(\tilde{\yvec}_1)}^2\leq\Lambda  
}} \,d \yvec_1 \,d\tilde{\yvec}_1\cdot P_{\Yvec_1|M}(\yvec_1|m) P_{\Yvec_1|M}(\tilde{\yvec}_1|\tm)
\nonumber\\
&\times \bigg[
\mathds{1}\left\{\yvec_1: g(\fvec'(m)+\fvec_1(\yvec_1)+\fvec'(\tm)+\fvec_1(\tilde{\yvec}_1))\neq m \right\}
+ \mathds{1}\left\{\yvec_1: g(\fvec'(m)+\fvec_1(\yvec_1)+\fvec'(\tm)+\fvec_1(\tilde{\yvec}_1))\neq \tm \right\} \bigg]
 \,.
\end{align}
As the sum in the square brackets is at least $1$ for all $\tm\neq m$, it follows that
\begin{align}
\err(q,\code) 
\geq&\frac{\kappa}{2\dM^2}\sum_{m=1}^{\dM}\sum_{\tm\neq m} \int_{\substack{(\yvec_1\,,\;\tilde{\yvec}_1) \,:\; 
 \frac{1}{n}\norm{
\fvec'(m)+\fvec_1(\yvec_1)}^2\leq\Lambda\,, \\
\frac{1}{n}\norm{
\fvec'(\tm)+\fvec_1(\tilde{\yvec}_1)}^2\leq\Lambda  
}} \,d \yvec_1 \,d\tilde{\yvec}_1\cdot P_{\Yvec_1|M}(y_1^n|m) P_{\Yvec_1|M}(\tilde{\yvec}_1|\tm)
\nonumber\\
\geq& \frac{\kappa}{4}\cdot \prob{ \begin{array}{l}
\frac{1}{n}\norm{\fvec'(M)+\fvec_1(\Yvec_1)}^2\leq\Lambda \,,\\
\frac{1}{n}\norm{\fvec'(\tM)+\fvec_1(\widetilde{\Yvec}_1)}^2\leq\Lambda \,,\; \tM\neq M
 \end{array}
  }
\,.
\end{align}
Then, recall that by (\ref{eq:RYgtSnLd}), the expectation of $\frac{1}{n}\norm{\fvec'(M)+\fvec_1(Y_1^n)}^2$ is strictly lower than $\Lambda$, and for a sufficiently large $n$, the conditional expectation of  $\frac{1}{n} || \fvec'(\tM)+\fvec_1(\widetilde{\Yvec}_1) ||^2$ given $\tM\neq M$ is also strictly lower than $\Lambda$. Thus, by Chebyshev's inequality, the probability of error is bounded from below by a positive constant. Following this contradiction, we deduce that if the code is reliable, then $\Lambda\leq (1+\alpha+2\rho\sqrt{\alpha})\plimit$.

It is left for us to show that for $\alpha$ and $\rho$ as defined in (\ref{eq:RYgRhoAdef}), we have that  $R<\mathsf{F}_G(\alpha,\rho)$ (see (\ref{eq:Fg})).
For a $(2^{nR},n,\eps_n^*)$ code, 
\begin{align}
P_{e|\svec}^{(n)}(\code)\leq \eps_n^* \,,
\label{eq:RYgConvCerrM}
\end{align}
for all $\svec\in\mathbb{R}^n$ with $\norm{\svec}^2\leq n\Lambda$.
Then, consider using the code $\code$ over the Gaussian relay channel $W^{\oq}_{Y,Y_1|X,X_1}$, specified by
\begin{align}
&Y_1=X''+Z \,, \nonumber\\
&Y=X'+X_1+\oS \,,
\end{align}
 where the sequence $\overline{\Svec}$ is i.i.d. $\sim\oq= 
\mathcal{N}(0,\Lambda-\delta)$.
First, we show that the code $\code$ is reliable for this channel, and then we show that $R<\mathsf{F}_G(\alpha,\rho)$.
Using the code $\code$ over the channel $W^{\oq}_{Y,Y_1|X,X_1}$, the probability of error is bounded by
\begin{align}
\err(\oq,\code)=\prob{\frac{1}{n}\norm{\overline{\Svec}}>\Lambda}+\int_{\svec:\frac{1}{n}\norm{\overline{\Svec}}\leq \Lambda}\, d\svec\cdot P_{e|\svec}^{(n)}(\code) \leq \eps_n^*+\eps_n^{**} \,,
\end{align}
where we have bounded the first term by $\eps_n^{**}$ using the law of large numbers  and the second term using (\ref{eq:RYgConvCerrM}), where $\eps_n^{**}\rightarrow 0$ as $n\rightarrow \infty$. 
Since  $W^{\oq}_{Y,Y_1|X,X_1}$ is a channel without a state, we can now show that $R<\mathsf{F}_G(\alpha,\rho)$ by following the lines of \cite{CoverElGamal:79p} and
\cite{ElGamalZahedi:05p}. By Fano's inequality and \cite[Lemma 4]{CoverElGamal:79p}, we have that
\begin{align}
R\leq&\,\frac{1}{n}\sum_{i=1}^n I_{\oq}(X_i',X_i'',X_{1,i};Y_i)+\eps_n  \,,\nonumber\\
R\leq&\,\frac{1}{n}\sum_{i=1}^n I_{\oq}(X_i',X_i'';Y_i,Y_{1,i}|X_{1,i})+\eps_n 
\,,
\label{eq:RYGorsUpperb1}
\end{align}
where $\oq=\mathcal{N}(0,\Lambda-\delta)$, $\Xvec'= \fvec'(M)$, $\Xvec''= \fvec''(M)$, $\Xvec_1=
\fvec_1(\Yvec_1)$, and $\eps_n\rightarrow 0$ as $n\rightarrow\infty$. 
For the Gaussian relay channel with SFD, we have the following Markov relations,
\begin{align}
&Y_{1,i}\Cbar X_i''\Cbar (X_i',X_{1,i},Y_{1,i}) \,, \label{eq:RYGorsY1markov}\\
&(X_i'',Y_{1,i})\Cbar (X_i',X_{1,i})\Cbar Y_{i} \,. \label{eq:RYGorsYmarkov} 
\end{align}
Hence,  by (\ref{eq:RYGorsYmarkov}), 
$
I_{\oq}(X_i',X_i'',X_{1,i};Y_i)=I_{\oq}(X_i',X_{1,i};Y_i)  
$. 
Moving to the second bound in the RHS of (\ref{eq:RYGorsUpperb1}), we follow the lines of \cite{ElGamalZahedi:05p}. Then, by the mutual information chain rule, we have
\begin{align}
&I_{\oq}(X_i',X_i'';Y_i,Y_{1,i}|X_{1,i}) \nonumber\\ =& I(X_i'';Y_{1,i}|X_{1,i})+I(X_i';Y_{1,i}|X_i'',X_{1,i})+
I_{\oq}(X_i',X_i'';Y_i|X_{1,i},Y_{1,i})
\nonumber\\
\stackrel{(a)}{=}& I(X_i'';Y_{1,i}|X_{1,i})+I_{\oq}(X_i',X_i'';Y_i|X_{1,i},Y_{1,i})
\nonumber\\
\stackrel{(b)}{=}& [H(Y_{1,i}|X_{1,i})-H(Y_{1,i}|X_i'')]+[H_{\oq}(Y_i|X_{1,i},Y_{1,i})-H_{\oq}(Y_i|X_i',X_{1,i})]
\nonumber\\
\stackrel{(c)}{\leq}& I_{q_1}(X_i'';Y_{1,i})+I(X_i';Y_i|X_{1,i})
\end{align}
where $(a)$ is due to (\ref{eq:RYGorsY1markov}), $(b)$ is due to (\ref{eq:RYGorsYmarkov}),  and
$(c)$ holds since conditioning reduces entropy.
Introducing a time-sharing random variable $K\sim \text{Unif}[1:n]$, which is independent of 
$\Xvec'$, $\Xvec''$, $\Xvec_1$, $\Yvec$, $\Yvec_1$, we have that 
\begin{align}
R-\eps_n\leq&\,  I_{\oq}(X_K',X_{1,K};Y_K|K)
\nonumber\\
R-\eps_n\leq&  I(X_K'';Y_{1,K}|K)+I_{\oq}(X_K';Y_K|X_{1,K},K) \,.
\label{eq:RYgConvIneqR}
\end{align}
Now, by the maximum differential entropy lemma (see \eg \cite[Theorem 8.6.5]{CoverThomas:06b}), 
\begin{align}
 I_{\oq}(X_K',X_{1,K};Y_K|K)
\leq&  \frac{1}{2} \log\left( \frac{ \E[(X_K'+X_{1,K})^2]+(\Lambda-\delta)}{\Lambda-\delta} \right)
=\frac{1}{2} \log\left(1+ \frac{ \alpha\plimit+\alpha_1\plimit_1+2\rho\sqrt{\alpha\plimit\cdot\alpha_1\plimit_1}}{\Lambda-\delta} \right)
 \,,
\end{align}
and
\begin{align}
I(X_K'';Y_{1,K}|K)+I_{\oq}(X_K';Y_K|X_{1,K},K)
\leq& \frac{1}{2}\log\left( \frac{ \E X_K''^2+\sigma^2}{\sigma^2}  \right) 
+ \frac{1}{2}\log\left( \frac{
\left[ 1-\frac{(\E(X_K'\cdot X_{1,K}))^2 }{\E X_K'^2 \cdot\E X_{1,K}^2}   \right]\E X_K'^2
+(\Lambda-\delta)}{\Lambda-\delta}  \right)
\nonumber\\
=& \frac{1}{2}\log\left( 1+\frac{ (1-\alpha)\plimit}{\sigma^2}  \right) 
+ \frac{1}{2}\log\left(1+ \frac{(1-\rho^2)\alpha\plimit}{\Lambda-\delta}  \right) \,,
\label{eq:RYgConvIneqR3}
\end{align}
where $\alpha$, $\alpha_1$ and $\rho$ are given by (\ref{eq:RYgRhoAdef}). 
Since $\delta>0$ is arbitrary, and $\alpha_1\leq 1$,  the proof follows from (\ref{eq:RYgConvIneqR})--(\ref{eq:RYgConvIneqR3}).
\qed

\section{Proof of Lemma~\ref{lemm:pRYcompoundDF}}
\label{app:pRYcompoundDF}
\subsection{Partial Decode-Forward Lower Bound}

We use superposition decode-forward coding, where the decoder uses joint typicality  with respect to a state type, which is ``close" to some  $q\in\Qset$.
For simplicity, assume that $C_1\leq\pRYdIRcompound$. The proof can be easily adjusted otherwise.
 Let $\delta>0$ be arbitrarily small, and
%
define a set of state types $\tQ$ by  
\begin{align}
\label{eq:pRYtQ}
\tQ=\{ \hP_{s^n} \,:\; s^n\in\Aset^{ \delta_1 
}(q) \;\text{ for some  $q\in\Qset
$}\, \} \,,
\end{align}
where 
\begin{align}
\label{eq:pRYdelta1}
\delta_1 \triangleq
\frac{\delta}{2\cdot |\Sset|} \,.
\end{align} 
Namely, $\tQ$ is the set of types that are $\delta_1$-close 
 to some state distribution $q(s)$ in $\Qset$. 
A code $\code$ for the compound relay channel 
 is constructed as follows. Each message is divided into two parts, $m=(m_1,m_2)$, where
$m_1\in [1:2^{nR_1}]$, $m_2\in [1:2^{nR_2}]$, and $R_1+R_2=R$ with $R_1\geq C_1$. 

\emph{Codebook Generation}: 
 Fix the distribution $P_{U,X}(u,x)$, and let 
\begin{align}
\label{eq:pRYCompoundDistUY}
P^q_{X,Y,Y_1|U}(x,y,y_1|u)= P_{X|U}(x|u)\cdot \sum_{s\in\Sset} q(s)  \prc(y,y_1|x,s) \,.
\end{align}
Generate $2^{n R_1}$ independent sequences $u^n(m_1)$, $m_1\in [1:2^{nR_1}]$, at random, each according to $\prod_{i=1}^n P_{U}(u_{i})$. 
Then, for every $m_1\in [1:2^{nR_1}]$, generate $2^{nR_2}$ sequences,
$x^n(m_1,m_2) \sim \prod_{i=1}^n P_{X|U}(x_i|u_{i}(m_1))$, 
where $m_2\in [1:2^{nR_2}]$, 
 conditionally independent given $u^n(m_1)$. 
Partition the set of indices $[1:2^{nR_1}]$ into  $2^{nC_1}$ bins of equal size,
$\Dset(\ell)=[ (\ell-1)2^{n(R_1-C_1)}+1: \ell 2^{n(R_1-C_1)} ]$ for $\ell\in [1:2^{nC_1}]$.

\emph{Encoding}:
To send $m=(m_1,m_2)$, transmit $x^n(m_1,m_2)$. 

\emph{Relay Encoding}:
The relay receives $y_{1}^n$, and finds a unique $\tm_1\in [1:2^{nR_1}]$ such that 
\begin{align}
(u^n(\tm_1),y_{1}^n)\in\tset(P_{U} P^{q}_{Y_1|U}) \,,\; 
\text{for some $q\in\tQ$}\,.
\end{align}
If there is none or there is more than one such, set $\tm_1=1$. 
The relay sends the associated bin index $\ell$, for which $\tm_1\in\Dset(\ell)$. 

\emph{Decoding}: The decoder receives $\ell$ and $y^n$. First, the decoder finds a unique $\hm_1\in
\Dset(\ell)$ such that
\begin{align}
 (u^n(\hm_1),y^n)\in\tset(P_{U} P^{q}_{Y|U}) \,,\;\text{for some $q\in\tQ$}\,.
\end{align}
If there is none, or more than one such $\hm_1\in[1:2^{nR_1}]$, declare an error.
Then, the decoder finds a unique $\hm_2\in[1:2^{nR_2}]$ such that
\begin{align}
 (u^n(\hm_1),x^n(\hm_1,\hm_2),y^n)\in\tset(P_{U,X} P^{q}_{Y|X}) \,,\; \text{for some $q\in\tQ$}\,.
\end{align}  
If there is none, or more than one such $\hm_2\in[1:2^{nR_2}]$, declare an error.
We note that using the set of types $\tQ$ instead of the original set of state distributions $\Qset$ alleviates the analysis, since
 $\Qset$ is not necessarily finite nor countable.

\emph{Analysis of Probability of Error}: 
Assume without loss of generality that the user sent  $(M_1,M_2)=(1,1)$, and let $q^*(s)\in\Qset$ denote the \emph{actual} state distribution chosen by the jammer.
The error event is bounded by the union of the  events
\begin{align}
\Eset_{1}=&\{ \tM_1\neq 1 \} \,,\; \Eset_{2}=\{ \hM_1\neq 1 \} \,,\; \Eset_{3}=\{ \hM_2\neq 1 \} \,.
\end{align}
Hence,  the probability of error is bounded by
\begin{align}
 \err(q,\code) \leq&  \prob{\Eset_1}+ \prob{\Eset_2\cap\Eset_1^c }+ \prob{\Eset_3\cap\Eset_1^c\cap \Eset_2^c}  \,,
\label{eq:pRYCompoundCerrBound}
\end{align}
 where the conditioning on $(M_1,M_2)=(1,1)$ is omitted for convenience of notation.

We begin with the probability of erroneous relaying, 
\begin{align}
\label{eq:pRYcompoundDFE1bound}
\prob{\Eset_1}\leq \prob{ \Eset_{1,1}}+\prob{ \Eset_{1,2}} \,,
\end{align}
where
\begin{align}
\Eset_{1,1}=& \{ (  U^n(1), Y_{1}^n  )\notin \tset(P_{U} P^{q'}_{Y_1|U})  \;\text{ for all $q'\in\tQ$} \} \nonumber\\
\Eset_{1,2}=& \{ 
(  U^n(m_1), Y_{1}^n  )\in \tset(P_{U} P^{q'}_{Y_1|U}) \,,\; 
\text{for some $m_1\neq 1$, $q'\in\tQ$} 
\} \,.
\label{eq:pRYcompuondE11}
\end{align}
Consider the first term on the RHS of (\ref{eq:pRYcompoundDFE1bound}).
We now claim that    the event $\Eset_{1,1}$ implies that 
$(  U^n(1), Y_{1}^n  )\notin \Aset^{\nicefrac{\delta}{2}}(P_{U} P^{q^*}_{Y_1|U})$.	
 Assume to the contrary that $\Eset_{1,1}$ holds, but $(  U^n(1), Y_{1}^n  )\in \Aset^{\nicefrac{\delta}{2}}(P_{U} P^{q^*}_{Y_1|U})$.  
Then, for a sufficiently large $n$, there exists a type $q'(s)$ such that 
$
|q'(s)-q^*(s)|\leq \delta_1 
$ 
 for all $s\in\Sset$, and by the definition in (\ref{eq:pRYtQ}), $q'\in\tQ$.  
We also have that 
\begin{align}
|P_{Y_1|U}^{q'}(y_1|u)-P_{Y_1|U}^{q^*}(y_1|u)|\leq |\Sset|\cdot \delta_1=\frac{\delta}{2} \,,
\end{align}
for all $u\in\Uset$ and $y_1\in\Yset_1$ (see (\ref{eq:pRYCompoundDistUY}) and (\ref{eq:pRYdelta1})). Hence, 
$(  U^n(1), Y_{1}^n  )\in \tset(P_{U} P^{q'}_{Y_1|U})$, and this contradicts the first assumption.
It follows that 
	\begin{align}
	\label{eq:pRYRYSllnRL}
	&\prob{ \Eset_{1,1}}	\leq \prob{(  U^n(1), Y_{1}^n  )\notin \Aset^{\nicefrac{\delta}{2}}(P_{U} P^{q^*}_{Y_1|U})  } \,,
	\end{align}
which  tends to zero exponentially as $n\rightarrow\infty$ by the law of large numbers and  Chernoff's bound.

	We move to the second term in the RHS of (\ref{eq:pRYcompoundDFE1bound}). 
	By the union of events bound and since the number of type classes in $\Sset^n$ is bounded by $(n+1)^{|\Sset|}$,   we have that  
\begin{align}
\label{eq:pRYSE2poly}
\prob{ \Eset_{1,2}} 
\leq& (n+1)^{|\Sset|}\cdot \sup_{q'\in\tQ} \prob{
(  U^n(m_1), Y_{1}^n  )\in \tset(P_{U} P^{q'}_{Y_1|U}) \;\text{ for some $m_1\neq 1$} 
} \nonumber\\
\leq& (n+1)^{|\Sset|}\cdot 2^{nR_1}  \cdot
 \sup_{q'\in\tQ}\left[  \sum_{u^n} P_{U^n}(u^n) \cdot \sum_{y_1^n \,:\; (u^n,y_1^n)\in \tset(P_{U} P^{q'}_{Y_1|U})} P_{Y_1^n}^{q^*}(y_1^n)
\right] \,,
\end{align}
where the last line follows since $U^n(m_1)$ is independent of $Y_{1}^n$, for every $m_1\neq 1$. 
 Let $y_1^n$ satisfy $(u^n,y_1^n)\in \tset(P_{U} P^{q'}_{Y_1|U})$. Then, $\,y_1^n\in\Aset^{\delta_2}(P_{Y_1}^{q'})$ with $\delta_2\triangleq |\Uset| \cdot\delta$. By Lemmas 2.6 and 2.7 in \cite{CsiszarKorner:82b},
\begin{align*}
P_{Y_1^n}^{q^*}(y_1^n)=2^{-n\left(  H(\hP_{y_1^n})+D(\hP_{y_1^n}||P_{Y_1}^{q^*})
\right)}\leq 2^{-n H(\hP_{y_1^n})}
\leq 2^{-n\left( H_{q'}(Y_1) -\eps_1(\delta) \right)} \,,
\end{align*}
where $\eps_1(\delta)\rightarrow 0$ as $\delta\rightarrow 0$. Therefore, by (\ref{eq:pRYSE2poly})
along with  \cite[Lemma 2.13]{CsiszarKorner:82b},
\begin{align}
& \prob{ \Eset_{1,2}}         																									
\leq
 \;(n+1)^{|\Sset|}\cdot \sup_{q'\in\Qset} 
2^{-n[ I_{q'}(U;Y_1) 
-R_1-\eps_2(\delta) ]} \label{eq:pRYSLexpCR} \,,
\end{align}
with $\eps_2(\delta)\rightarrow 0$ as $\delta\rightarrow 0$.
We now have by (\ref{eq:pRYcompoundDFE1bound}) that $\prob{\Eset_1}$  tends to zero exponentially as $n\rightarrow\infty$, provided that $R_1<\inf_{q'\in\Qset} I_{q'}(U;Y_1)
-\eps_2(\delta)$.

As for the erroneous decoding of $M_1$ at the receiver,  define the events,
\begin{align}
\Eset_{2,1}=& \{ (  U^n(1), Y^n  )\notin \tset(P_{U} P^{q'}_{Y|U})  \;\text{ for all $q'\in\tQ$} \} \nonumber\\
\Eset_{2,2}=& \{ (  U^n(m_1), Y^n  )\in \tset(P_{U} P^{q'}_{Y|U}) \,,\; 
\text{for some $m_1\in\Dset(L)$, $m_1\neq 1$, $q'\in\tQ$} \} \,,
\end{align}
where $L$ is the index sent by the relay.
 Then,
\begin{align}
\prob{\Eset_2\cap\Eset_1^c}\leq& 
\prob{\Eset_{2,1}} 
+\prob{\Eset_{2,2}\cap\Eset_1^c} \,.
\label{eq:pRYcompoundE2b}
\end{align}
By similar arguments to those used above, we have that $\prob{\Eset_{2,1}}$
  tends to zero exponentially as $n\rightarrow\infty$ by the law of large numbers and Chernoff's bound.
Moving to the second term on the RHS of (\ref{eq:pRYcompoundE2b}),
observe that given $\Eset_1^c$, the relay sends the index $L$ for which $M_1\in \Dset(L)$, \ie the decoder receives $L=1$.
 Thus, by similar arguments to those used for the bound on $\prob{ \Eset_{1,2}}$, we have that
\begin{align}
\label{eq:pRYSE22poly}
\prob{\Eset_{2,2}\cap\Eset_1^c} 
\leq& (n+1)^{|\Sset|}\cdot 2^{n(R_1-C_1)}  \cdot
 \sup_{q'\in\tQ}\left[  \sum_{u^n} P_{U^n}(u^n) \cdot \sum_{y^n \,:\; (u^n,y^n)\in \tset(P_{U} P^{q'}_{Y|U})} P_{Y^n}^{q^*}(y^n)
\right]
\nonumber\\
\leq&
 \;(n+1)^{|\Sset|}\cdot \sup_{q'\in\Qset} 
2^{-n[ I_{q'}(U;Y) 
-R_1+C_1-\eps_3(\delta) ]} 
 \,,
\end{align}
with $\eps_3(\delta)\rightarrow 0$ as $\delta\rightarrow 0$. By (\ref{eq:pRYcompoundE2b}), we have that
the second term in the RHS of (\ref{eq:pRYCompoundCerrBound}) tends to zero exponentially as $n\rightarrow\infty$, provided that 
$R_1<\inf_{q'\in\Qset} I_{q'}(U;Y)+C_1-\eps_3(\delta)$

Moving to the error event for $M_2$, define
\begin{align}
\Eset_{3,1}=& \{ (  U^n(\hM_1), X^n(\hM_1,1), X_{1,b}(\hM_1),
 Y^n  )\notin \tset(P_{U,X} P^{q'}_{Y|X}) \,,\; 
\text{for all $q'\in\tQ$} \} \nonumber\\
\Eset_{3,2}=& \{ (  U^n(\hM_1), X^n(\hM_1,m_2),  Y^n  )\in \tset(P_{U,X} P^{q'}_{Y|X}) \,,\; 
\text{for some $m_2\neq 1$, $q'\in\tQ$} \} \,.
\end{align}
Given $\Eset_2^c$, we have that $\hM_1=1$. Then, by similar arguments to those used above,
\begin{align}
&\prob{\Eset_3\cap\Eset_1^c\cap \Eset_2^c} 
\leq  \prob{\Eset_{3,1}\cap\Eset_1^c\cap \Eset_2^c}
 +\prob{\Eset_{3,2}\cap\Eset_1^c\cap \Eset_2^c} \nonumber\\
\leq& e^{-a_0 n}
+(n+1)^{|\Sset|}\cdot \sup_{q'\in\Qset} 
 \sum_{m_2\neq 1}\cprob{(  U^n(1), X^n(1,m_2|1), Y^n  )\in \tset(P_{U,X} P^{q'}_{Y|X})}{\Eset_1^c}
\nonumber\\
\leq&e^{-a_0 n}+  (n+1)^{|\Sset|}\cdot \sup_{q'\in\Qset} 
2^{-n[ I_{q'}(X;Y|U) -R_2-\eps_4(\delta) ]}
\label{eq:pRYcompoundE3b}
\end{align}
where $a_0>0$ and $\eps_4(\delta)\rightarrow 0$ as $\delta\rightarrow 0$. 
The second inequality holds by the law of large numbers and Chernoff's bound, and the last inequality holds as $X^n(1,m_2)$ is conditionally independent of $Y^n$ given $U^n(1)$ for every $m_2\neq 1$.
Thus, the third term on the RHS of (\ref{eq:pRYCompoundCerrBound})  tends to zero exponentially as $n\rightarrow\infty$, provided that $R_2< \inf_{q'\in\Qset} I_{q'}(X;Y|U)-\eps_4(\delta)$.
Eliminating $R_1$ and $R_2$, we conclude that the probability of error, averaged over the class of the codebooks, exponentially decays to zero  as $n\rightarrow\infty$, provided that $R<\pRYdIRcompound$. Therefore, there exists a $(2^{nR},n,\eps)$ deterministic code, for a sufficiently large $n$.
\qed

\subsection{Cutset Upper Bound}
This is a straightforward consequence of the cutset bound in \cite[Proposition 1]{Kim:07c}.
Assume to the contrary that there exists an achievable rate $R>\pRYICcompound$. Then, for some 
$q^*(s)$ in the closure of $\Qset$, 
\begin{align}
\label{eq:pcutsetCompound1}
R>\max_{p(x)} 
\min \left\{ I_{q^*}(X;Y)+C_1 \,,\; I_{q^*}(X;Y,Y_1)
\right\} \,.
\end{align}
By the achievability assumption, we have that for every $\eps>0$ and sufficiently large $n$, there exists a $(2^{nR},n)$ random code $\code^\Gamma$ such that $\err(q,\code)\leq\eps$ for every i.i.d. state distribution
 $q\in\Qset$, and in particular for $q^*$. This holds even if $q^*$ is in the closure of $\Qset$ but not in $\Qset$ itself, since $\err(q,\code)$ is continuous in $q$.
Consider using this code over a primitive relay channel $W_{Y,Y_1|X}$ without a state, where
$
W_{Y,Y_1|X}(y,y_1|x)=\sum_{s\in\Sset} q^*(s) \prc(y,y_1|x,s) 
$. 
It follows that the rate $R$ as in (\ref{eq:pcutsetCompound1}) can be achieved over the relay channel $W_{Y,Y_1|X,X_1}$, in contradiction to \cite{Kim:07c}. We deduce that the assumption is false, and $R>\pRYICcompound$ cannot be achieved.
\qed

\section{Proof of Corollary~\ref{coro:pRYcompoundDeg}}
\label{app:pRYcompoundDeg}
This is a straightforward consequence of Lemma~\ref{lemm:pRYcompoundDF}, which states that
the capacity of the compound primitive relay channel is bounded by 
$\pRYdIRcompound\leq\pRYCcompound\leq \pRYICcompound$.
Thus, if $\prc$ is reversely strongly degraded,  $I_q(X;Y,Y_1)=I_q(X;Y)$, and
the bounds coincide by the minimax theorem \cite{sion:58p}, \cf (\ref{eq:pRYICcompound}) and (\ref{eq:pRYcompoundDirectTran}).
Similarly, if $\prc$ is strongly degraded, then  $I_q(X;Y,Y_1)=I(X;Y_1)$, and by (\ref{eq:pRYICcompound}) and (\ref{eq:pRYcompoundFullDF}), 
\begin{align}
\pRYICcompound=& \min_{q(s)\in\Qset}  \max_{p(x)} 
\min \left\{ I_q(X;Y)+C_1 \,,\; I(X;Y_1)
\right\} \ \,, 
\label{eq:pRYdegradedICcompound}
\\
\pRYdIRcompound=&   \max_{p(x)} \min_{q(s)\in\Qset} 
\min \left\{ I_q(X;Y)+C_1 \,,\; I(X;Y_1)
\right\} \,.
\label{eq:pRYdegradedIRcompound}
\end{align} 
Observe that $\min \left\{ I_q(X;Y)+C_1 \,,\; I(X;Y_1)
\right\}$ is concave in $p(x)$ and quasi-convex in $q(s)$ (see \eg \cite[Section 3.4]{BoydVandenbergh:04b}), hence  the bounds (\ref{eq:pRYdegradedICcompound}) and (\ref{eq:pRYdegradedIRcompound}) coincide by the minimax theorem \cite{sion:58p}.
\qed

\section{Proof of Theorem~\ref{theo:pRYmain}}
\label{app:pRYmain}
Consider a primitive AVRC $\pavrc$.


\subsection{Partial Decode Forward Lower Bound}
The proof is based on Ahlswede's RT \cite{Ahlswede:86p},
stated below. Let $h:\Sset^n\rightarrow [0,1]$ be a given function. If $
\sum_{s^n\in\Sset^n} \qn(s^n)h(s^n)\leq \alpha_n 
$, 
 for 
all $ \qn(s^n)=\prod_{i=1}^n q(s_i)$, 
$q\in\pSpace(\Sset)$, with 
 $\alpha_n\in(0,1)$,
then,
\begin{align}
\label{eq:pRYRTresC}
\frac{1}{n!} \sum_{\pi\in\Pi_n} h(\pi s^n)\leq \beta_n \,,\quad\text{for all $s^n\in\Sset^n$} \,,
\end{align}
where $\Pi_n$ is the set of all $n$-tuple permutations $\pi:\Sset^n\rightarrow\Sset^n$, and 
$\beta_n=(n+1)^{|\Sset|}\cdot\alpha_n$. 

Let $R<\pRYdIRavc$. According to 
Lemma~\ref{lemm:pRYcompoundDF}, there exists a  $(2^{nR},n,e^{-2\theta n})$ 
 code for the compound primitive relay channel 
$\pRYcompoundP$,  for some $\theta>0$ and sufficiently large $n$. 
Given such a  code $\code=(f,f_1,g)$ for 
 $\pavrc^{\pSpace(\Sset)}$,  
 we have that  
$\sum_{s^n\in\Sset^n} q(s^n) h(s^n) \leq e^{-2\theta n} 
$, 
for all $q\in\pSpace(\Sset)$, where
\begin{align}
&h(s^n)=  \frac{1}{2^{nR}}\sum_{m\in [1:2^{nR}]}  
\sum_{(y^n,y_1^n) : g(y^n,f_1(y_1^n))\neq m } W_{Y^n,Y_1^n|X^n,S^n}(y^n,y_1^n| 
f(m),s^n  )\,.
\label{eq:pRYrch1}
\end{align}
Hence, applying  Ahlswede's RT, we have that for a sufficiently large $n$,
\begin{align}
&\frac{1}{n!} \sum_{\pi\in\Pi_n} h(\pi s^n)\leq 
(n+1)^{|\Sset|}e^{-2\theta n} 
\leq e^{-\theta n}  \,,\;\text{for all $s^n\in\Sset^{n}$} \,.
\label{eq:pRYdetErrC}
\end{align} 

On the other hand, for every $\pi\in\Pi_n$, 
\begin{align}
h(\pi s^n)=& \frac{1}{2^{nR}}\sum_{m\in [1:2^{nR}]}  \sum_{(y^n,y_1^n) : g(y^n,f_1(y_1^n))\neq m } W_{Y^n|X^n,S^n}
(y^n,y_1^n| f(m),\pi s^n  )\nonumber\\
 \stackrel{(a)}{=}& \frac{1}{2^{nR}}\sum_{m\in [1:2^{nR}]}  
\sum_{(y^n,y_1^n) : g(\pi y^n,f_1(\pi y_1^n))\neq m } W_{Y^n|X^n,S^n}
(\pi y^n| f(m),\pi s^n  )\nonumber\\
 \stackrel{(b)}{=}& \frac{1}{2^{nR}}\sum_{m\in [1:2^{nR}]}  
\sum_{(y^n,y_1^n) : g(\pi y^n,f_1(\pi y_1^n))\neq m } W_{Y^n|X^n,S^n}
( y^n| \pi^{-1} f(m),  s^n  )  \,,
\label{eq:pRYcerrpi1}
\end{align}
where in $(a)$ we change 
 the order of summation over $(y^n,y_{1}^n)$, and $(b)$ holds because the  channel is memoryless.
%
Then, consider the $(2^{nR},n)$ random code $\code^\Pi$, specified by 
$
f_{\pi}(m)= \pi^{-1} f(m) 
$, $
f_{1,\pi}(y_{1}^n)= f_{1}(\pi y_{1}^n) 
$, and 
$
g_{\pi}(y^n,\ell)=g(\pi y^n,\ell)
$, 
for $\pi\in\Pi_n$,
with a uniform distribution $\mu(\pi)=\frac{1}{n!}$. 
 By (\ref{eq:pRYcerrpi1}), 
 the probability of error of this random code 
is bounded by 
$
\err(\qn,\code^\Pi)\leq e^{-\theta n} 
$, 
for every $\qn(s^{n})\in\pSpace(\Sset^{n})$. 
\qed 

\subsection{Cutset Upper Bound}
The proof immediately follows from Lemma~\ref{lemm:pRYcompoundDF},
since the random code capacity of the primitive AVRC is bounded by the random code capacity  of the compound primitive relay channel, \ie
$
\pRYrCav \leq \pRYrCcompoundP 
$. 
\qed 

\section{Proof of Lemma~\ref{lemm:pRYcorrTOdetC}}
\label{app:pRYcorrTOdetC}
 We follow the lines of \cite{Ahlswede:78p}, with the required adjustments. We use the random code constructed in the proof of  Theorem~\ref{theo:pRYmain}. Let $R<\pRYrCav$, and consider the case where
$C_1>0$ and the marginal sender-relay  AVC has positive capacity, \ie
\begin{align}
\label{eq:pRYrpos}
\opC(\avc_1)>0 \,,
\end{align}
(see (\ref{eq:pRYmarginAVCs})).
By Ahlswede's Elimination Technique \cite{Ahlswede:78p}, 
  for every $\eps>0$ and sufficiently large $n$, 
 there exists a $(2^{nR},n,\eps)$ random  code  
$
\code^\Gamma=\big(\mu(\gamma)=\frac{1}{k},\Gamma=[1:k],\{\code_\gamma \}_{\gamma\in \Gamma}\big) 
$, 
where
$\code_\gamma=(\encn_\gamma,\enc_{1,\gamma},\dec_{\gamma})$, for $\gamma\in\Gamma$, 
 and 
$k=|\Gamma|\leq n^2 $. 
Following (\ref{eq:pRYrpos}),  we have that for every $\eps_1>0$ and sufficiently large $\nu$, the code index $\gamma\in [1:k]$ can be sent through the relay channel $W_{Y_1|X,S}$ using a $(2^{\nu\bR},\nu,\eps_1)$ deterministic code 
$ 
\code_{\text{i}}=(\tf^{\nu},\gnu)  
$, where $\bR>0$. Since $k$ is at most polynomial, 
 the encoder can reliably convey $\gamma$ to the relay with a negligible blocklength, \ie
$ 
\nu=o(n) 
$. 

Now, consider a code 
 formed by the concatenation of $\code_{\text{i}}$ as a prefix to a corresponding code in the code collection $\{\code_\gamma\}_{\gamma\in\Gamma}$. 
That is, the encoder first sends the index $\gamma$ to the relay, and then it sends the message $m\in [1:2^{nR}]$ to the receiver. Specifically, the encoder first transmits  $\tx^\nu=\tf^{\nu}(\gamma)$ in order to convey the index $\gamma$ to the relay. At the end of this transmission, the relay uses the first $\nu$ symbols it received to  estimate the code index as $\hgamma=\tg(\ty_1^{\nu})$. Since $C_1>0$ and  $k$ is at most polynomial in $n$,
 the relay can reliably convey its estimation $\hgamma$ to the receiver with a negligible blocklength
$ 
\nu'=o(n) 
$. 

Then, the message $m$ is transmitted by the codeword $x^n=\enc_\gamma(m)$. The decoder uses the estimated index $\hgamma$ received from the relay, and the message is estimated by  
$\widehat{m}=$ $g_{\hgamma}(y^n)$.  
 By the union of events bound, the probability of error 
 is then bounded by $\eps_c=\eps+\eps_1$, for every joint distribution in $\pSpace(\Sset^{\nu+\nu'+n})$. 
That is, the concatenated code is a $(2^{(\nu+\nu'+n)\tR_n},\nu+\nu'+n,\eps_c)$ code over the primitive AVRC $\pavrc$, where 
  the blocklength is $n+o(n)$, and  
the rate   $\tR_{n}=\frac{n}{\nu+\nu'+n}\cdot R$  approaches $R$ as $n\rightarrow \infty$. 
\qed

\section{Proof of Theorem~\ref{theo:pRYmainDbound}}
\label{app:pRYmainDbound}
Consider part 1.
If $W_{Y_1|X,S}$ is non-symmetrizable, then  $\opC(\avc_1)>0$ by \cite[Theorem 1]{CsiszarNarayan:88p}. Hence, by Lemma~\ref{lemm:pRYcorrTOdetC},
$\pRYCavc=\pRYrCav$, and by Theorem~\ref{theo:pRYmain}, $\pRYdIRavc \leq \pRYCavc\leq\pRYrICav$. 
Part 2 and part 3 follow from part 1 and Corollary~\ref{coro:pRYavrcDeg}.
Part 4 follows by the arguments in \cite[Appendix G]{PeregSteinberg:17a2}. 
\qed
\end{appendices}

\printbibliography
 
\end{document}